\documentclass[%
 reprint,
superscriptaddress,
nofootinbib,
 amsmath,amssymb,
 aps,
]{revtex4-2}
\usepackage{physics}
\usepackage{bbold}
\usepackage{amsmath}
\usepackage{latexsym}
\usepackage{graphicx}
\usepackage{amsthm}
\usepackage{hyperref}
\usepackage{lipsum}
\usepackage[english]{babel}
\usepackage{slashed}
\usepackage[compat=1.1.0]{tikz-feynman}
\usepackage{appendix}
\usepackage{subfig}
\usepackage{multirow}
\usepackage{mathtools}
\usepackage{subfig}
\usepackage{pgfplots}
\pgfplotsset{compat=1.15}
\usepackage{mathrsfs}
\usetikzlibrary{arrows}
\definecolor{ffffff}{rgb}{1,1,1}
\definecolor{qqqqff}{rgb}{0,0,1}
\definecolor{qqwuqq}{rgb}{0,0.39215686274509803,0}
\definecolor{zzttqq}{rgb}{0.6,0.2,0}
\definecolor{ttqqqq}{rgb}{0.2,0,0}
\definecolor{qqttqq}{rgb}{0,0.2,0}
\definecolor{qqccqq}{rgb}{0,0.8,0}
\usepackage{tikz-cd}
\usepackage{amsmath}
\usepackage{amsbsy}

\usetikzlibrary[patterns]
\usepackage{csquotes}
\usepackage{ragged2e}

\theoremstyle{plain}

\newtheorem{theorem}{Theorem}

\newtheorem{lemma}{Lemma}

\tikzset{
pattern size/.store in=\mcSize, 
pattern size = 5pt,
pattern thickness/.store in=\mcThickness, 
pattern thickness = 0.3pt,
pattern radius/.store in=\mcRadius, 
pattern radius = 1pt}
\makeatletter
\pgfutil@ifundefined{pgf@pattern@name@_b8864pxwu}{
\pgfdeclarepatternformonly[\mcThickness,\mcSize]{_b8864pxwu}
{\pgfqpoint{0pt}{0pt}}
{\pgfpoint{\mcSize+\mcThickness}{\mcSize+\mcThickness}}
{\pgfpoint{\mcSize}{\mcSize}}
{
\pgfsetcolor{\tikz@pattern@color}
\pgfsetlinewidth{\mcThickness}
\pgfpathmoveto{\pgfqpoint{0pt}{0pt}}
\pgfpathlineto{\pgfpoint{\mcSize+\mcThickness}{\mcSize+\mcThickness}}
\pgfusepath{stroke}
}}
\makeatother

\tikzset{
pattern size/.store in=\mcSize, 
pattern size = 5pt,
pattern thickness/.store in=\mcThickness, 
pattern thickness = 0.3pt,
pattern radius/.store in=\mcRadius, 
pattern radius = 1pt}
\makeatletter
\pgfutil@ifundefined{pgf@pattern@name@_phbkm4h6k}{
\pgfdeclarepatternformonly[\mcThickness,\mcSize]{_phbkm4h6k}
{\pgfqpoint{0pt}{0pt}}
{\pgfpoint{\mcSize+\mcThickness}{\mcSize+\mcThickness}}
{\pgfpoint{\mcSize}{\mcSize}}
{
\pgfsetcolor{\tikz@pattern@color}
\pgfsetlinewidth{\mcThickness}
\pgfpathmoveto{\pgfqpoint{0pt}{0pt}}
\pgfpathlineto{\pgfpoint{\mcSize+\mcThickness}{\mcSize+\mcThickness}}
\pgfusepath{stroke}
}}
\makeatother

\tikzset{
pattern size/.store in=\mcSize, 
pattern size = 5pt,
pattern thickness/.store in=\mcThickness, 
pattern thickness = 0.3pt,
pattern radius/.store in=\mcRadius, 
pattern radius = 1pt}
\makeatletter
\pgfutil@ifundefined{pgf@pattern@name@_bjy7bfjd0}{
\pgfdeclarepatternformonly[\mcThickness,\mcSize]{_bjy7bfjd0}
{\pgfqpoint{0pt}{-\mcThickness}}
{\pgfpoint{\mcSize}{\mcSize}}
{\pgfpoint{\mcSize}{\mcSize}}
{
\pgfsetcolor{\tikz@pattern@color}
\pgfsetlinewidth{\mcThickness}
\pgfpathmoveto{\pgfqpoint{0pt}{\mcSize}}
\pgfpathlineto{\pgfpoint{\mcSize+\mcThickness}{-\mcThickness}}
\pgfusepath{stroke}
}}
\makeatother

\tikzset{
pattern size/.store in=\mcSize, 
pattern size = 5pt,
pattern thickness/.store in=\mcThickness, 
pattern thickness = 0.3pt,
pattern radius/.store in=\mcRadius, 
pattern radius = 1pt}
\makeatletter
\pgfutil@ifundefined{pgf@pattern@name@_1xz8bldgc}{
\pgfdeclarepatternformonly[\mcThickness,\mcSize]{_1xz8bldgc}
{\pgfqpoint{0pt}{-\mcThickness}}
{\pgfpoint{\mcSize}{\mcSize}}
{\pgfpoint{\mcSize}{\mcSize}}
{
\pgfsetcolor{\tikz@pattern@color}
\pgfsetlinewidth{\mcThickness}
\pgfpathmoveto{\pgfqpoint{0pt}{\mcSize}}
\pgfpathlineto{\pgfpoint{\mcSize+\mcThickness}{-\mcThickness}}
\pgfusepath{stroke}
}}
\makeatother
\tikzset{every picture/.style={line width=0.75pt}} %set default line width to 0.75pt   

\tikzset{
pattern size/.store in=\mcSize, 
pattern size = 5pt,
pattern thickness/.store in=\mcThickness, 
pattern thickness = 0.3pt,
pattern radius/.store in=\mcRadius, 
pattern radius = 1pt}
\makeatletter
\pgfutil@ifundefined{pgf@pattern@name@_7uoi8qg3w}{
\pgfdeclarepatternformonly[\mcThickness,\mcSize]{_7uoi8qg3w}
{\pgfqpoint{0pt}{0pt}}
{\pgfpoint{\mcSize+\mcThickness}{\mcSize+\mcThickness}}
{\pgfpoint{\mcSize}{\mcSize}}
{
\pgfsetcolor{\tikz@pattern@color}
\pgfsetlinewidth{\mcThickness}
\pgfpathmoveto{\pgfqpoint{0pt}{0pt}}
\pgfpathlineto{\pgfpoint{\mcSize+\mcThickness}{\mcSize+\mcThickness}}
\pgfusepath{stroke}
}}
\makeatother

\tikzset{
pattern size/.store in=\mcSize, 
pattern size = 5pt,
pattern thickness/.store in=\mcThickness, 
pattern thickness = 0.3pt,
pattern radius/.store in=\mcRadius, 
pattern radius = 1pt}
\makeatletter
\pgfutil@ifundefined{pgf@pattern@name@_74qutwv71}{
\pgfdeclarepatternformonly[\mcThickness,\mcSize]{_74qutwv71}
{\pgfqpoint{0pt}{0pt}}
{\pgfpoint{\mcSize+\mcThickness}{\mcSize+\mcThickness}}
{\pgfpoint{\mcSize}{\mcSize}}
{
\pgfsetcolor{\tikz@pattern@color}
\pgfsetlinewidth{\mcThickness}
\pgfpathmoveto{\pgfqpoint{0pt}{0pt}}
\pgfpathlineto{\pgfpoint{\mcSize+\mcThickness}{\mcSize+\mcThickness}}
\pgfusepath{stroke}
}}
\makeatother
\tikzset{every picture/.style={line width=0.75pt}}
 
\tikzset{
pattern size/.store in=\mcSize, 
pattern size = 5pt,
pattern thickness/.store in=\mcThickness, 
pattern thickness = 0.3pt,
pattern radius/.store in=\mcRadius, 
pattern radius = 1pt}
\makeatletter
\pgfutil@ifundefined{pgf@pattern@name@_f6j18z8qd}{
\pgfdeclarepatternformonly[\mcThickness,\mcSize]{_f6j18z8qd}
{\pgfqpoint{0pt}{0pt}}
{\pgfpoint{\mcSize+\mcThickness}{\mcSize+\mcThickness}}
{\pgfpoint{\mcSize}{\mcSize}}
{
\pgfsetcolor{\tikz@pattern@color}
\pgfsetlinewidth{\mcThickness}
\pgfpathmoveto{\pgfqpoint{0pt}{0pt}}
\pgfpathlineto{\pgfpoint{\mcSize+\mcThickness}{\mcSize+\mcThickness}}
\pgfusepath{stroke}
}}
\makeatother
\tikzset{every picture/.style={line width=0.75pt}} %set default line width to 0.75pt

\tikzset{
pattern size/.store in=\mcSize, 
pattern size = 5pt,
pattern thickness/.store in=\mcThickness, 
pattern thickness = 0.3pt,
pattern radius/.store in=\mcRadius, 
pattern radius = 1pt}
\makeatletter
\pgfutil@ifundefined{pgf@pattern@name@_avbgigq0y}{
\pgfdeclarepatternformonly[\mcThickness,\mcSize]{_avbgigq0y}
{\pgfqpoint{0pt}{0pt}}
{\pgfpoint{\mcSize+\mcThickness}{\mcSize+\mcThickness}}
{\pgfpoint{\mcSize}{\mcSize}}
{
\pgfsetcolor{\tikz@pattern@color}
\pgfsetlinewidth{\mcThickness}
\pgfpathmoveto{\pgfqpoint{0pt}{0pt}}
\pgfpathlineto{\pgfpoint{\mcSize+\mcThickness}{\mcSize+\mcThickness}}
\pgfusepath{stroke}
}}
\makeatother

\tikzset{
pattern size/.store in=\mcSize, 
pattern size = 5pt,
pattern thickness/.store in=\mcThickness, 
pattern thickness = 0.3pt,
pattern radius/.store in=\mcRadius, 
pattern radius = 1pt}
\makeatletter
\pgfutil@ifundefined{pgf@pattern@name@_3ioveo5rx}{
\pgfdeclarepatternformonly[\mcThickness,\mcSize]{_3ioveo5rx}
{\pgfqpoint{0pt}{0pt}}
{\pgfpoint{\mcSize+\mcThickness}{\mcSize+\mcThickness}}
{\pgfpoint{\mcSize}{\mcSize}}
{
\pgfsetcolor{\tikz@pattern@color}
\pgfsetlinewidth{\mcThickness}
\pgfpathmoveto{\pgfqpoint{0pt}{0pt}}
\pgfpathlineto{\pgfpoint{\mcSize+\mcThickness}{\mcSize+\mcThickness}}
\pgfusepath{stroke}
}}
\makeatother

\tikzset{
pattern size/.store in=\mcSize, 
pattern size = 5pt,
pattern thickness/.store in=\mcThickness, 
pattern thickness = 0.3pt,
pattern radius/.store in=\mcRadius, 
pattern radius = 1pt}
\makeatletter
\pgfutil@ifundefined{pgf@pattern@name@_9g507j6sk}{
\pgfdeclarepatternformonly[\mcThickness,\mcSize]{_9g507j6sk}
{\pgfqpoint{0pt}{0pt}}
{\pgfpoint{\mcSize+\mcThickness}{\mcSize+\mcThickness}}
{\pgfpoint{\mcSize}{\mcSize}}
{
\pgfsetcolor{\tikz@pattern@color}
\pgfsetlinewidth{\mcThickness}
\pgfpathmoveto{\pgfqpoint{0pt}{0pt}}
\pgfpathlineto{\pgfpoint{\mcSize+\mcThickness}{\mcSize+\mcThickness}}
\pgfusepath{stroke}
}}
\makeatother

\tikzset{
pattern size/.store in=\mcSize, 
pattern size = 5pt,
pattern thickness/.store in=\mcThickness, 
pattern thickness = 0.3pt,
pattern radius/.store in=\mcRadius, 
pattern radius = 1pt}
\makeatletter
\pgfutil@ifundefined{pgf@pattern@name@_yt1t5kpfm}{
\pgfdeclarepatternformonly[\mcThickness,\mcSize]{_yt1t5kpfm}
{\pgfqpoint{0pt}{0pt}}
{\pgfpoint{\mcSize+\mcThickness}{\mcSize+\mcThickness}}
{\pgfpoint{\mcSize}{\mcSize}}
{
\pgfsetcolor{\tikz@pattern@color}
\pgfsetlinewidth{\mcThickness}
\pgfpathmoveto{\pgfqpoint{0pt}{0pt}}
\pgfpathlineto{\pgfpoint{\mcSize+\mcThickness}{\mcSize+\mcThickness}}
\pgfusepath{stroke}
}}
\makeatother
\tikzset{every picture/.style={line width=0.75pt}} %set default line width to 0.75pt  

\begin{document}

\title{Mixture equivalence principles and post-quantum theories of gravity}
\author{Samuel Fedida}
\affiliation{Centre for Quantum Information and Foundations, DAMTP, Centre for Mathematical Sciences, University of Cambridge, Wilberforce Road, Cambridge CB3 0WA, UK}
\author{Adrian Kent}
\affiliation{Centre for Quantum Information and Foundations, DAMTP, Centre for Mathematical Sciences, University of Cambridge, Wilberforce Road, Cambridge CB3 0WA, UK}
\affiliation{Perimeter Institute for Theoretical
	Physics, 31 Caroline Street North, Waterloo, ON N2L 2Y5, Canada.}
\date{\today}

\begin{abstract}
    We examine the mixture equivalence principle (MEP), which states that proper and improper mixed states with the same density matrix are always experimentally indistinguishable, and a weaker version, which states that this is sometimes true in gravity theories. We point out that M{\o}ller-Rosenfeld semiclassical gravity violates the weak MEP and that nonlinear extensions of quantum mechanics violate the MEP. We further demonstrate that modifications of the Born rule in quantum theory also typically violate the MEP. We analyse such violations in the context of thermal baths, where proper and improper thermal states induce different physical situations.  This has significant implications in the context of black hole physics. We argue that M{\o}ller-Rosenfeld semiclassical gravity is not the semiclassical limit of quantum gravity in the context of black hole spacetimes, even in the presence of $N\gg1$ matter fields.  
\end{abstract}

\maketitle

\section{Introduction}

In this paper we consider a class of extensions of quantum theory 
that have the same mathematical description of physical states
as quantum theory, but allow non-standard measurements on
these states.
Our prime motivating example is 
M{\o}ller-Rosenfeld semiclassical gravity, which we generally hereafter refer to simply as semiclassical gravity for brevity, and which is defined as taking the semiclassical Einstein field equations 
\begin{equation}
    \label{Semiclassical Field Equations}
    G_{\mu \nu} = \frac{8\pi G}{c^4} \expval{\hat{T}_{\mu\nu}}
\end{equation}
to be fundamental. 

Contradictions with observation \cite{Page1981} and apparent paradoxes \cite{Kibble1980,Duff1980,Gisin1990} arise in a theory coupling classical gravity to quantum matter in this way. However, these can at least partly be addressed by adding other hypotheses (as discussed in e.g. \cite{Tilloy2016,tilloy_binding_2018,Tilloy_2019,Grossardt2022}). For example, if the quantum state of matter undergoes objective collapses that normally prevent superpositions of macroscopically distinct matter states, or if we suppose that semiclassical gravity is an approximation with a restricted domain of validity, then there is not necessarily an immediate contradiction with experiment or observation. 
Moreover, models of this type can be constructed without introducing superluminal signalling \cite{Kent2005,Kent2025measurement}.   
So, there continues to be some motivation for exploring theories involving some form of semiclassical gravity and investigating their features, implications, relationships and potential problems.

Another alternative to quantizing gravity is to try to construct consistent theories that combine a classical gravitational field with a linear stochastic quantum state evolution (e.g. \cite{Tilloy_2019,Oppenheim2023}).  We do not consider these possibilities in this paper, which focusses on the properties
and implications of nonlinear theories.    

The semiclassical Einstein field equations (\ref{Semiclassical Field Equations}) are also often taken as a limit of quantum gravity.
In this context, it is not
usually argued that the above paradoxes
and contradictions with observation are problematic, and indeed
they are not much discussed.   One of the points of this paper is to resolve
what might initially seem a sociological puzzle by noting that
the relevant ``semiclassical limit of quantum gravity'' is not 
identical to M{\o}ller-Rosenfeld semiclassical gravity, though
closely related.  

The role of mixtures in the relevant theories is key to our discussions.
There are two kinds of state mixing in quantum theory: improper mixing, which arises as a description of a subsystem obtaining from tracing out a degree of freedom of an entangled multipartite system, and proper mixing, which 
arises from (classical) statistical considerations of an ensemble of states. 
A priori, one might expect that these could be physically and observationally distinct notions.   In particular, this might seem natural if pure states represent or imply distinct ontologies, as is the case in several versions of quantum theory.  In this case, proper and improper mixed states are ontologically distinct, and it seems logically possible that they might always be empirically distinguishable.  

In fact, of course, in standard quantum theory, both proper and improper mixed states are described by a density operator $\rho \in \mathscr{D}(\mathcal{H})$, which encapsulates all the physically accessible information obtainable from acting (only) on the relevant \mbox{(sub-)system}.
In particular, the expectation value of any operator $\mathcal{O}$ representing an observable of the (sub-)system is $\Tr(\mathcal{O}\rho)$, whether $\rho$ represents a proper or an improper mixed state.    Every experimental prediction of quantum theory can be expressed in terms of expectation values.  Hence, if a proper mixed state and an improper mixed state are represented by the same density matrix, 
$\rho_{\text{proper}} = \rho_{\text{improper}}$, then these two states are experimentally indistinguishable. 

However, we are interested in extensions of quantum theory that retain its mathematical description of physical states but allow non-standard measurements.  Thus, the states of a system $S$ are represented by rays in a Hilbert space $\mathcal{H}_S$ and those of a composite system $S_1 + \ldots + S_N$ by rays in 
$\mathcal{H}_{S_1} \otimes \ldots \otimes \mathcal{H}_{S_n}$.
For simplicity we restrict our discussion to countable probabilistic ensembles. 
We assume throughout that the preparation statistics of any ensemble are independent of the outcome statistics of any later measurements (i.e.~the probability weighting $\{p_i\}_i$ chosen for the $\psi_i$ does not influence the statistics of measurements on each individual $\psi_i$).    For completeness and later reference we spell out the
implication in the following trivial lemma.   
\begin{lemma}
    \label{Expval proper mixed states}
    Let 
    \begin{enumerate}
        \item $\{\ket{\psi_i}\}_i \in \mathcal{H}$ be an ensemble of pure states with associated probabilities $\{p_i\}_i, \, \sum_i p_i = 1$, where $\mathcal{H}$ is the Hilbert space of the theory, 
        \item $\rho_\text{proper} = \sum_i p_i \psi_i \in \mathscr{D}(\mathcal{H})$ be a proper mixed state, where we wrote the pure state $\psi_i = \ket{\psi_i}\bra{\psi_i} \in \mathscr{D}(\mathcal{H})$,
        \item $\mathcal{O} \in \mathcal{L}(\mathcal{H})$ be an observable, i.e. $\mathcal{O}^\dagger = \mathcal{O}$, where the expectation value of measurements of $\mathcal{O}$ on a pure $\psi_i$ is denoted $\expval{\mathcal{O}}_{\psi_i}$.
    \end{enumerate}
    Then, 
    \begin{equation}
        \expval{\mathcal{O}}_{\rho_\text{proper}} = \sum_i p_i \expval{\mathcal{O}}_{\psi_i} \, .
    \end{equation}
\end{lemma}

\begin{proof}
    Any measurement of observable $\mathcal{O}$ conducted on a pure state $\psi_i$ has expectation $\expval{\mathcal{O}}_{\psi_i}$. Thus, undertaking measurements of an observable $\mathcal{O}$ conducted on an ensemble of pure states $\{\ket{\psi_i}\}_i$ with associated preparation probabilities $\{p_i\}_i$ has expectation
    \begin{equation}
        \sum_i p_i \expval{\mathcal{O}}_{\psi_i} \, , 
    \end{equation}
    since the preparation statistics and the outcome statistics are independent.
\end{proof}

We also suppose that the system $S_1$ (or later $S$) we consider may always be taken to be part of a 
composite system $\ket{\Psi} \in \mathcal{H}_{S_1} \otimes \mathcal{H}_{S_2}$,
where $\mathcal{H}_{S_2}$ is infinite-dimensional. 

Given an ensemble $\{\ket{\psi_i}\}_i \in \mathcal{H}_S$ of pure states with associated probabilities $\{p_i\}_i, \, \sum_i p_i = 1$, 
we define the corresponding proper mixed state
\begin{equation}
    \rho_{\text{proper}} = \sum_i p_i \ket{\psi_i} \bra{ \psi_i } \, . 
\end{equation}
   
Given a pure state $\ket{\Psi} \in \mathcal{H}_{S_1} \otimes \mathcal{H}_{S_2}$ we define the corresponding improper mixed state
for $S_1$ as
\begin{equation}
    \rho_{\text{improper}} = \Tr_{\mathcal{H}_{S_2}} ( \ket{\Psi} \bra{ \Psi } ) \, . 
\end{equation}
Because our extensions include measurements not allowed in quantum theory, it may be possible to distinguish a state of $S_1$ prepared as a probabilistic ensemble from the state representing it as a subsystem
of a composite system in an entangled pure state, even when the 
corresponding mixed states are mathematically equal, 
i.e. when $\rho_{\text{proper}} = \rho_{\text{improper}}$.  

There are various interesting possibilities. 
One is that the extended theory satisfies what we will call the Mixture Equivalence Principle (MEP) \cite{Kent2025measurement}. 
\begin{quote}

  {\bf MEP:}  \qquad   Let $\{\ket{\psi_i}\}_i \in \mathcal{H_S}$ be an ensemble of pure states with associated probabilities $\{p_i\}_i, \, \sum_i p_i = 1$, where $\mathcal{H_S}$ is the Hilbert space representing a system $S$, and $\rho_{\text{proper}}$ be the corresponding proper mixed state.   Let $\ket{ \Psi } \in \mathcal{H_S} \otimes  \mathcal{H_A}$ be a state in the Hilbert space representing the combined system $S+A$, and $\rho_{\text{improper}}$ be the corresponding improper mixed state of $S$.   Suppose $\rho_{\text{proper}}= \rho_{\text{improper}}$.   
  Then no experiment on $S$ can distinguish these two cases.
  
\end{quote}

By transitivity, the MEP also implies that no experiment can distinguish any pair of proper mixtures, or any pair of improper mixtures, represented by the same density matrix.   
The MEP is well understood to be a fundamental feature of standard 
quantum theory, and in particular of the physics of decoherence.  When a quantum system interacts with an environment or a macroscopic apparatus or observer, the (improper) reduced density matrix of the quantum system, assuming unitary evolution, asymptotically tends to the (proper) density matrix defined by an ensemble of pure states with associated Born probabilities in some natural basis defined by the interaction. Hence the statistics of experiments conducted on the quantum system are identical to those one could get from preparing an ensemble of pure states with these probabilities. 

Another is what we will call the Weak Mixture Equivalence Principle (WMEP).
\begin{quote}
 {\bf WMEP:}  \qquad  There exists an ensemble of experimentally distinguishable pure states $\{\ket{\psi_i}\}_i \in \mathcal{H_S}$ with associated probabilities $\{p_i\}_i, \, \sum_i p_i = 1$,  and 
 a state $\ket{ \Psi } \in \mathcal{H_S} \otimes  \mathcal{H_A}$, 
 such that $\rho_{\text{proper}}= \rho_{\text{improper}}$ and 
 such that no experiment on $S$ can distinguish these two cases. 
\end{quote}

Related to these is the so-called Purification Principle (PP), which was defined for more general theories but makes sense in the restricted context we consider here. 
\cite{PhysRevA.81.062348,Galley2018}. 
\begin{quote}
{\bf PP:}  \qquad  For any ensemble of pure states $\{\ket{\psi_i}\}_i \in \mathcal{H_S}$ with associated probabilities $\{p_i\}_i, \, \sum_i p_i = 1$, there exists a state $\ket{ \Psi } \in \mathcal{H_S} \otimes  \mathcal{H_A}$, such that $\rho_{\text{proper}}= \rho_{\text{improper}}$ and such that no experiment on $S$ can distinguish these two cases. \footnote{In the original formulation \cite{PhysRevA.81.062348}, the state $\ket{\Psi}$ is unique up to the action of local unitaries on $\mathcal{H_A}$.} 
\end{quote}

Clearly, for the class of theories we consider here,
the MEP implies the PP, which implies the WMEP.    
As we will see, even the WMEP does not hold in 
semiclassical gravity.   

\section{Gravitational violations of Mixture Equivalence Principles}

\label{Proper and improper mixed states}

In the context of theories of gravity, it is natural to 
define gravitational versions of the above principles.
For example, we define the Gravitational Mixture Equivalence Principle (GMEP).
\begin{quote}
{\bf GMEP:}  \qquad  
Let $\{\ket{\psi_i}\}_i \in \mathcal{H_S}$ be an ensemble of pure states with associated probabilities $\{p_i\}_i, \, \sum_i p_i = 1$, where $\mathcal{H_S}$ is the Hilbert space representing a system $S$, and $\rho_{\text{proper}}$ be the corresponding proper mixed state.   Let $\ket{ \Psi } \in \mathcal{H_S} \otimes  \mathcal{H_A}$ be a state in the Hilbert space representing the combined system $S+A$, and $\rho_{\text{improper}}$ be the corresponding improper mixed state of $S$.   Suppose $\rho_{\text{proper}}= \rho_{\text{improper}}$.   Then no measurable gravitational effect of $S$ can distinguish these two cases.
\end{quote}

It is logically possible that any of the above principles 
could hold for all experiments not involving gravity, but that
their gravitational versions could fail.
It is also possible that the gravitational version of any of these
principles could fail only via violations of the original principle. In this case, any gravitational effect that distinguishes proper and improper mixed states arises because of the (perhaps amplified) gravitational effects of an apparatus carrying out measurements that do not involve gravity and that distinguish proper and improper mixed states. 

We focus first on the (G)MEP.   
Note that it is possible that the (G)MEP could be violated by running a single experiment involving some classical or quantum observable $\mathcal{O}$ -- in which case we talk about a \textit{one-shot} (G)MEP violation -- but that the asymptotic repetition of that experiment involving that same observable $\mathcal{O}$ does not yield any statistical difference between measurements on proper and improper mixed states. If there is a statistical difference, such that 
\begin{equation}
    \expval{\mathcal{O}}_{\rho_\text{proper}} \neq \expval{\mathcal{O}}_{\rho_\text{improper}} \, ,
\end{equation}
where the expectation value may be classical or quantum, we say that the (G)MEP is \textit{statistically} violated with measurements of $\mathcal{O}$ for that experiment. Of course, statistical violations of the (G)MEP with some observable $\mathcal{O}$ for a given experiment necessarily imply one-shot violations of the (G)MEP with that same observable in the experiment, but the converse is not true. The (G)MEP may thus either hold, be violated one-shot, or be violated both one-shot and statistically with respect to an observable $\mathcal{O}$ in an experiment\footnote{Note here that, although the terminology may suggest the opposite intuition, one-shot (G)MEP violations are weaker than statistical (G)MEP violations. That is, a statistical (G)MEP violation for a given observable implies a one-shot (G)MEP violation for that observable. The converse is untrue, as seen in the case of semiclassical gravity in the following section.}.

\subsection{Semiclassical gravity and one-shot GMEP violations}

\label{Semiclassical gravity section}

Semiclassical gravity violates the GMEP, as we now review (see discussions in e.g. \cite{tilloy_binding_2018,Kent2021b,Grossardt2022}). 
Consider the thought experiment depicted in Figure \ref{fig:MEP violation}.
\begin{figure}[b!]
    \centering
    \includegraphics[width=1 \columnwidth]{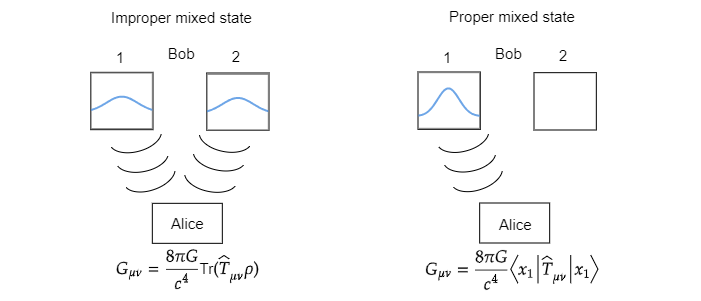}
    \caption{Gravitational violation of the MEP in semiclassical gravity}
    \label{fig:MEP violation}
\end{figure}
Alice and Bob agree on an inertial reference frame in which they will remain at agreed fixed separation during the experiment, and on the experiment's protocol. There are now two versions.

In the version of the experiment depicted on the right of Figure \ref{fig:MEP violation}, 
Bob then prepares a particle of mass $m$ whose wavefunction is localised around points $\mathbf{x}_1$ and $\mathbf{x}_2$ in one of two boxes 1 and 2. Bob chooses the box randomly, with probabilities $\abs{a_i}^2, \, i=1,2$ for each (say, by tossing a weighted classical coin whose weights are also known by Alice). Alice observes a test mass $M$ initially at spacetime point $(\mathbf{y},t')$. She initially has no information about which of the two boxes the particle is in so, for her, the quantum state of Bob's mass is initially in the proper mixture 
\begin{equation}\label{aliceproper}
    \rho_{\text{proper}} = \abs{a_1}^2 \ket{\mathbf{x}_1} \bra{\mathbf{x}_1} + \abs{a_2}^2\ket{\mathbf{x}_2} \bra{\mathbf{x}_2} \, .
\end{equation}
However, $\rho_{\text{proper}}$ does not ultimately describe the actual state of Bob from Alice's perspective. Once gravitational information (travelling at light speed) affects Alice's mass, it does so sourced by Bob's particle located at a definite position. If not, Bob and Alice would make different predictions. From Bob's perspective, his mass has a definite position, so the gravitational interaction is not defined by
\begin{equation}
    G_{\mu\nu} = \frac{8\pi G}{c^4} \Tr(\hat{T}_{\mu\nu} \rho_{\text{proper}})
\end{equation}
but rather by
\begin{equation}
    \label{Semiclassical Einstein Field Equations Proper Case}
    G_{\mu\nu} = \frac{8\pi G}{c^4} \expval{\hat{T}_{\mu\nu}}{\mathbf{x}_i}
\end{equation}
for the outcome $i=1$ or $2$ (with no summation on the $i$ index). 

This is true in the standard understanding of both quantum gravity and semiclassical gravity.  For example, in the nonrelativistic limit of semiclassical gravity \eqref{Semiclassical Field Equations} we can define the semiclassical Newtonian potential to follow Poisson's equation
\begin{equation}
    \label{Semiclassical Newtonian potential}
    \Delta \Phi(\mathbf{x}) = 4\pi G \expval{\hat{\rho}(\mathbf{x})}
\end{equation}
where $\hat{\rho}(\mathbf{x})$ is the mass density operator of the quantum matter. For a single particle of mass $m$, the semiclassical gravitational field at some point $\mathbf{y}$ (Alice's position) can be found from Poisson's equation \eqref{Semiclassical Newtonian potential} and is given by
\begin{equation}
    \label{Semiclassical gravitational field proper}
    \Phi(\mathbf{y})_{(\text{proper})} = \frac{-Gm}{\abs{\mathbf{x}_i-\mathbf{y}}}
\end{equation}
for the outcome $i$. 

Thus, as soon as Alice enters the future light cone of Bob's random box choice, she can obtain information about that choice, and updates the proper mixture to the relevant pure state.  

In the version on the left, on the other hand, Bob keeps his random choice indeterminate at the quantum level by preparing an entangled state
\begin{equation}\label{bobentangled}
a_1 \ket{0} \ket{\mathbf{x}_1} + {a_2} \ket{1} \ket{\mathbf{x}_2}  \, ,
\end{equation}
where $\ket{0},\ket{1}$ are orthogonal states of an ancilla qubit, 
the mass $m$ is initially in an improper mixture 
\begin{equation}
 \rho_{\text{improper}} = \abs{a_1}^2\ket{\mathbf{x}_1} \bra{\mathbf{x}_1} + \abs{a_2}^2\ket{\mathbf{x}_2} \bra{\mathbf{x}_2} \, . 
\end{equation}
   According to semiclassical gravity,  Alice's mass follows the semiclassical Einstein field equations
\begin{equation}
    G_{\mu\nu} = \frac{8\pi G}{c^4} \Tr(\hat{T}_{\mu\nu} \rho_{\text{improper}}) \, , 
\end{equation}
which differs from \eqref{Semiclassical Einstein Field Equations Proper Case}, even though $\rho_{\text{improper}} = \rho_{\text{proper}}$. In the nonrelativistic limit, the semiclassical gravitational field is now given by
\begin{equation}
    \label{Semiclassical gravitational field improper}
    \Phi(\mathbf{y})_{(\text{improper})} = -Gm\Big(\frac{\abs{a_1}^2}{\abs{\mathbf{x}_1-\mathbf{y}}} + \frac{\abs{a_2}^2}{\abs{\mathbf{x}_2-\mathbf{y}}}\Big) \, .
\end{equation}
Determining the classical gravitational field provides information about the quantum matter state by giving an estimate of the $\abs{a_i}^2$ for $i=1,2$. Since $\Tr(\rho)=1$, one reading of the gravitational field suffices to estimate both $\abs{a_1}^2$ and $\abs{a_2}^2$.
The dynamics of her test mass give her no information about the entangled state other than $\rho_{\text{improper}}$.
This suggests that her observations should not alter the entangled state, and indeed semiclassical gravity postulates no alteration.  

Accordingly to semiclassical gravity, Alice can thus perform a single (\textit{one-shot}) Cavendish experiment to distinguish between the case where Bob has a proper mixed state \eqref{Semiclassical gravitational field proper} and that where he has an improper mixed state \eqref{Semiclassical gravitational field improper} since $\Phi(\mathbf{y})_{(\text{proper})} \neq \Phi(\mathbf{y})_{(\text{improper})}$. Semiclassical gravity thus allows observers to distinguish between proper and improper mixed states. In other words, the GMEP is violated in semiclassical gravity.
In contrast, it holds in standard quantum gravity, according to which
Alice's measurement of the gravitational field would effectively 
collapse Bob's entangled state into one of the components of \ref{bobentangled}, leading to the same probability distribution of outcomes obtained from the proper mixture \ref{aliceproper}.  

In the standard account of semiclassical gravity, the MEP holds for quantum matter experiments in which spacetime backreactions are undetectable, but not for general experiments involving gravity: i.e. the MEP fails only via the failure of the GMEP.  
However, it is not clear from this discussion whether there actually is a fully consistent semiclassical gravity theory: we consider this further below.

\subsection{Mixtures of mixtures and GMEP violations}

\label{Beyond proper and improper mixed states}

Consider, as above, a thought experiment involving Alice and Bob with now four boxes, as depicted in Figure \ref{fig:four boxes}. Alice and Bob agree on an inertial reference frame in which they will remain at agreed fixed separation during the experiment, and on the experiment's protocol. There are now three versions. 

\begin{figure*}
    \centering
    \includegraphics[width = 1.9 \columnwidth]{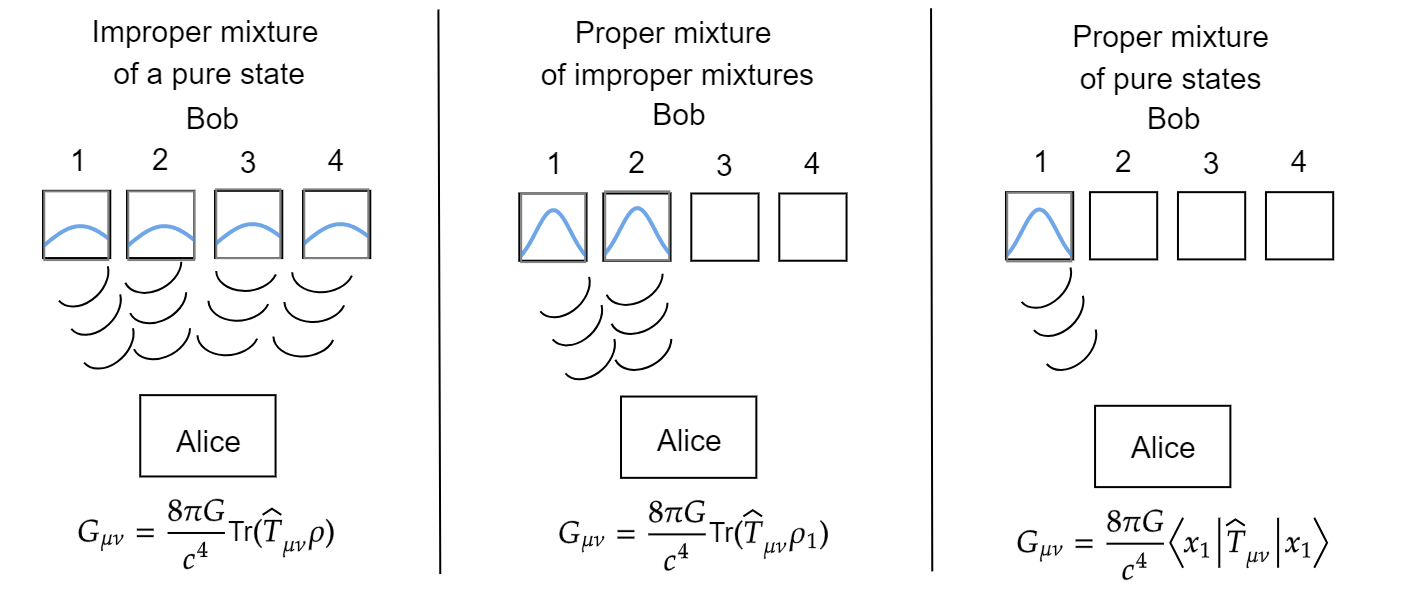}
    \caption{Gravitational violation of the MEP in semiclassical gravity, which distinguishes between three types of mixtures: improper mixtures of pure states, proper mixtures of improper mixtures and proper mixtures of pure states.}
    \label{fig:four boxes}
\end{figure*}

\subsubsection{Proper mixture of pure states} 

In the version of the experiment depicted on the far right of Figure \ref{fig:four boxes}, Bob prepares a particle of mass $m$ whose wavefunction is localised around points $\mathbf{x}_1$, $\mathbf{x}_2$, $\mathbf{x}_3$ and $\mathbf{x}_4$ in one of four boxes 1,2,3 and 4. Bob chooses the box randomly, with probabilities $p_1 \abs{a_1}^2, p_1 \abs{a_2}^2, p_2 \abs{a_3}^2, p_2 \abs{a_4}^2$ respectively (say, by tossing a weighted four-sided classical dice whose weights are also known by Alice). Alice observes a test mass $M$ initially at spacetime point $(\mathbf{y},t')$. She initially has no information about which of the two boxes the particle is in so, for her, the quantum state of Bob's mass is initially in the proper mixture 
\begin{multline}\label{aliceproper2}
    \rho_{\text{proper}} = p_1 \abs{a_1}^2 \ket{\mathbf{x}_1} \bra{\mathbf{x}_1} + p_1 \abs{a_2}^2\ket{\mathbf{x}_2} \bra{\mathbf{x}_2} \\ + p_2 \abs{a_3}^2 \ket{\mathbf{x}_3} \bra{\mathbf{x}_3} + p_2 \abs{a_4}^2\ket{\mathbf{x}_4} \bra{\mathbf{x}_4} \, .
\end{multline}
However, $\rho_{\text{proper}}$ does not ultimately describe the actual state of Bob from Alice's perspective. Once gravitational information (travelling at light speed) affects Alice's mass, it does so sourced by Bob's particle located at a definite position. If not, Bob and Alice would make different predictions. From Bob's perspective, his mass has a definite position, so the gravitational interaction is not defined by
\begin{equation}
    G_{\mu\nu} = \frac{8\pi G}{c^4} \Tr(\hat{T}_{\mu\nu} \rho_{\text{proper}})
\end{equation}
but rather by
\begin{equation}
    \label{Semiclassical Einstein Field Equations Proper Case2}
    G_{\mu\nu} = \frac{8\pi G}{c^4} \expval{\hat{T}_{\mu\nu}}{\mathbf{x}_i}
\end{equation}
for the outcome $i=1$ or $2$ or $3$ or $4$ (with no summation on the $i$ index). Again, this is true in the standard understanding of both quantum gravity and semiclassical gravity. Thus, as soon as Alice enters the future light cone of Bob's random box choice, she can obtain information about that choice, and updates the proper mixture to the relevant pure state. 

\subsubsection{Single pure entangled state}

In the version on the far left of Figure \ref{fig:four boxes}, on the other hand, Bob keeps his random choice indeterminate at the quantum level by preparing a pure entangled state
\begin{equation}\label{bobentangled2}
\sqrt{p_1} (a_1  \ket{0} \ket{\mathbf{x}_1} + a_2 \ket{1} \ket{\mathbf{x}_2}) + \sqrt{p_2} (a_3  \ket{2} \ket{\mathbf{x}_3} + a_4  \ket{3} \ket{\mathbf{x}_4}) \, ,
\end{equation}
where $\ket{0},\ket{1},\ket{2},\ket{3}$ are orthogonal states of an ancilla qudit, 
the mass $m$ is initially in an improper mixture 
\begin{multline}
 \rho_{\text{improper}} = p_1 \abs{a_1}^2 \ket{\mathbf{x}_1} \bra{\mathbf{x}_1} + p_1 \abs{a_2}^2\ket{\mathbf{x}_2} \bra{\mathbf{x}_2} \\ + p_2 \abs{a_3}^2 \ket{\mathbf{x}_3} \bra{\mathbf{x}_3} + p_2 \abs{a_4}^2\ket{\mathbf{x}_4} \bra{\mathbf{x}_4} \, . 
\end{multline}
According to semiclassical gravity,  Alice's mass follows the semiclassical Einstein field equations
\begin{equation}
    \label{eqn:scg impr}
    G_{\mu\nu} = \frac{8\pi G}{c^4} \Tr(\hat{T}_{\mu\nu} \rho_{\text{improper}}) \, .
\end{equation}
 This differs from \eqref{Semiclassical Einstein Field Equations Proper Case2}, even though $\rho_{\text{improper}} = \rho_{\text{proper}}$. This violates the GMEP, as seen previously.

\subsubsection{Mixtures of mixtures}

Another scenario is described in the centre of Figure \ref{fig:four boxes}.  Bob first tosses a weighted classical coin (whose weights are also known by Alice), and prepares, with probability $p_1$, the entangled state 
\begin{equation}
    \ket{\psi_1} = a_1 \ket{0} \ket{\mathbf{x}_1} + {a_2} \ket{1} \ket{\mathbf{x}_2}
\end{equation}
and with probability $p_2 = 1 - p_1$ the entangled state
\begin{equation}
    \ket{\psi_2} = a_3 \ket{0} \ket{\mathbf{x}_3} + {a_4} \ket{1} \ket{\mathbf{x}_4}
\end{equation}
where $\ket{0},\ket{1}$ are orthogonal states of an ancilla qubit living in some Hilbert space $\mathcal{H}_1$. The mass $m$ is initially in a mixture
\begin{multline}
 \rho_{\text{prop}\vert\text{imp}} = p_1 \abs{a_1}^2 \ket{\mathbf{x}_1} \bra{\mathbf{x}_1} + p_1 \abs{a_2}^2\ket{\mathbf{x}_2} \bra{\mathbf{x}_2} \\ + p_2 \abs{a_3}^2 \ket{\mathbf{x}_3} \bra{\mathbf{x}_3} + p_2 \abs{a_4}^2\ket{\mathbf{x}_4} \bra{\mathbf{x}_4} \, . 
\end{multline} 
According to semiclassical gravity,  Alice's mass follows the semiclassical Einstein field equations
\begin{equation}
    \label{Semiclassical Einstein Field Equations Improper Case2}
    G_{\mu\nu} = \frac{8\pi G}{c^4} \Tr(\hat{T}_{\mu\nu} \rho_j) \, , 
\end{equation}
where $j = 1$ or $2$ for \begin{equation}
    \rho_1 = \Tr_{\mathcal{H}_1}[\ket{\psi_1}\bra{\psi_1}] = \abs{a_1}^2\ket{\mathbf{x}_1} \bra{\mathbf{x}_1} + \abs{a_2}^2\ket{\mathbf{x}_2} \bra{\mathbf{x}_2}
\end{equation} and
\begin{equation}
    \rho_2 = \Tr_{\mathcal{H}_1}[\ket{\psi_2}\bra{\psi_2}] = \abs{a_3}^2\ket{\mathbf{x}_3} \bra{\mathbf{x}_3} + \abs{a_4}^2\ket{\mathbf{x}_4} \bra{\mathbf{x}_4} \, .
\end{equation}
This differs from both \eqref{Semiclassical Einstein Field Equations Proper Case2} and \eqref{eqn:scg impr} although $\rho_{\text{prop}\vert\text{imp}} = \rho_{\text{improper}} = \rho_{\text{proper}}$. Further note that, had Bob prepared instead the entangled state
\begin{equation}
    \ket{\psi_3} =  a_1 \ket{0} \ket{\mathbf{x}_1} + a_3 \ket{1} \ket{\mathbf{x}_3}
\end{equation}
with probability $p_1$ and
\begin{equation}
    \ket{\psi_4} = a_2 \ket{0} \ket{\mathbf{x}_2} + a_4 \ket{1} \ket{\mathbf{x}_4}
\end{equation}
with probability $p_2 = 1-p_1$, then the initial mixture would be the same as all the above, although Alice's mass would follow the semiclassical Einstein field equations
\begin{equation}
    G_{\mu\nu} = \frac{8\pi G}{c^4} \Tr(\hat{T}_{\mu\nu} \rho_k) \, , 
\end{equation}
for $k=3$ or $4$ for
\begin{equation}
    \rho_3 = \Tr_{\mathcal{H}_1}[\ket{\psi_3}\bra{\psi_3}] = \abs{a_1}^2\ket{\mathbf{x}_1} \bra{\mathbf{x}_1} + \abs{a_3}^2\ket{\mathbf{x}_3} \bra{\mathbf{x}_3}
\end{equation} and
\begin{equation}
    \rho_4 = \Tr_{\mathcal{H}_1}[\ket{\psi_4}\bra{\psi_4}] = \abs{a_2}^2\ket{\mathbf{x}_2} \bra{\mathbf{x}_2} + \abs{a_4}^2\ket{\mathbf{x}_4} \bra{\mathbf{x}_4} \, .
\end{equation}
This differs from either of the above.  Similarly, one can think of an entangled state between boxes $1-4$ and $2-3$, and generalise such considerations to any finite number of boxes. Thus, one needs to distinguish carefully  constructions of mixtures in (G)MEP-violating theories.

We note here that there is no operational distinction in this example between an improper mixture of a proper mixture $\Tr_{\mathcal{H}_1}[\{p_1,\ket{\psi_1};p_2,\ket{\psi_2}\}]$, and a proper mixture of improper mixtures $\{p_1,\Tr_{\mathcal{H}_1}[\ket{\psi_1}];p_2,\Tr_{\mathcal{H}_1}[\ket{\psi_2}]\}$.   In extensions of quantum theory where improper mixtures are not built from partial traces -- but rather from a nonlinear map -- this need not hold. One could also postulate non-standard measurement rules that treat these two cases separately.

\subsection{Semiclassical gravity violates the GWMEP}

We define the Gravitational Weak Mixture Equivalence Principle (GWMEP) as follows.

\begin{quote}
{\bf GWMEP:}  \qquad  
There exists an ensemble of experimentally distinguishable pure states $\{\ket{\psi_i}\}_i \in \mathcal{H_S}$ with associated probabilities $\{p_i\}_i, \, \sum_i p_i = 1$,  and 
 a state $\ket{ \Psi } \in \mathcal{H_S} \otimes  \mathcal{H_A}$, 
 such that $\rho_{\text{proper}}= \rho_{\text{improper}}$ and 
 such that no measurable gravitational effect
 of $S$ can distinguish these two cases. 
 \end{quote}

Consider a system $S$ in a proper mixture defined by an ensemble $E$ of experimentally distinguishable pure states $\{\ket{\psi_i}\}_i \in \mathcal{H_S}$ with associated probabilities $\{p_i\}_i, \, \sum_i p_i = 1$ and $p_i>0$.   All experimental measurements ultimately involve
measurements of mass densities in localised regions.  (This is true however one models measurement outcomes: for example, as determined by the location of an apparatus pointer, ink on paper, or local densities of chemical species in an observer's brain.)    So, the experimental distinguishability of 
the $\{\ket{\psi_i}\}_i$ means that there must be an ancilla 
$A$ initially in some reference state $\ket{0}_A$, a unitary operation $U$ acting on $\mathcal{H_A}\otimes\mathcal{H_S}  $ and an ideal position measurement $\{P_j \}$
acting on $\mathcal{H_A} \otimes \mathcal{H_S}$ such that 
the sets $I_j = \{E_j\}$ of the expectation values of the $P_j$ on the states $U (\ket{0}_A \otimes \ket{\psi_i}  )$
are distinct.   

Now suppose that $S$ is a subsystem of $S+S'$.   
A joint state $\ket{\Psi} \in \mathcal{H_S} \otimes \mathcal{H_{S'}}$
is indistinguishable from the ensemble $E$ by standard quantum
operations (not involving gravity) on $S$ if and only if
\begin{equation}
\ket{\Psi} = \sum_i (p_i )^{1/2} \ket{\psi_i}_S \otimes \ket{i}_{S'}  \, , 
\end{equation}
where $\{ \ket{i}_{S'} \}$ are an orthonormal set in $\mathcal{H_{S'}}$. 
Introducing the ancilla $A$ and applying the operation
$U \otimes I_{S'}$ to $\ket{0}_A \otimes\ket{\Psi}  $,
we obtain 
\begin{equation}
U \otimes I_{S'} ( \ket{0}_A  \otimes \ket{\Psi}) = \sum_i (p_i )^{1/2}  U (\ket{0}_A \otimes \ket{\psi_i}) \otimes \ket{i}_{S'} \, . 
\end{equation}
By construction, this is a superposition of states with different
mass density expectation values.    
As noted above, if this state describes the matter degrees of 
freedom, semiclassical gravity implies a gravitational
field obtained from the weighted average of these mass density
distributions.
In contrast, for the proper mixture $E$, semiclassical gravity
implies a gravitational field obtained from one of the mass
density distributions.  Since the distributions are distinct,
semiclassical gravity thus makes distinct predictions for the proper
and improper mixtures.   This gives a one-shot violation of the 
GWMEP.  

\subsection{Nonlinear quantum mechanics and statistical MEP violations}

\label{Nonlinear QM violates the MEP statistically}

We now argue that extensions of quantum mechanics for which the time-evolution operator of the theory is not linear violate the MEP.

We consider theories with a general time evolution law 
defined by operators $T_{\rho}(t,t_0) : \mathscr{D}(\mathcal{H}) \to \mathscr{D}(\mathcal{H})$, so that the state $\rho \equiv \rho(t_0)$ at
time $t_0$ evolves to $\rho(t) = T_{\rho}(t,t_0)[\rho(t_0)] \in \mathscr{D}(\mathcal{H})$, where $t>t_0$.   
We say the theory is nonlinear if there exists $\{ p_i \}, \{ \rho_i \} , t_0, t $
such that the mixed state 
\begin{equation} \rho(t_0) = \sum_{i=1}^N p_i \rho_i(t_0)
\end{equation}
at time $t_0$ 
evolves at time $t$ to 
\begin{equation}
    T_{\rho}(t,t_0)[\rho(t_0) ] \neq \sum_{i=1}^N p_i T_{\rho_i}(t,t_0)[\rho_i(t_0)] \, .  
\end{equation}

\begin{lemma}
    \label{Lemma Nonlinear qm violates the MEP}
    Let 
    \begin{enumerate}
        \item $\{\ket{\psi_i(t_0)}\}_i \in \mathcal{H}$ be an ensemble of pure states with associated probabilities $\{p_i\}_i, \, \sum_i p_i = 1$, where $\mathcal{H}$ is the Hilbert space of the theory, 
        \item $\rho(t_0) = \sum_i p_i \psi_i(t_0) \in \mathscr{D}(\mathcal{H})$ where we wrote the pure state $\psi_i(t_0) = \ket{\psi_i(t_0)}\bra{\psi_i(t_0)} \in \mathscr{D}(\mathcal{H})$,
        \item $T_\rho(t,t_0) : \mathscr{D}(\mathcal{H}) \to \mathscr{D}(\mathcal{H})$ be a time-evolution operator for the state $\rho(t_0)$ such that $\rho(t) = T_\rho(t,t_0)[\rho(t_0)] \in \mathscr{D}(\mathcal{H})$, and likewise we write $\psi_i(t) = T_{\psi_i}(t,t_0)[\psi_i(t_0)]$,
        \item $\mathcal{O} \in \mathcal{L}(\mathcal{H})$ be an observable, i.e. $\mathcal{O}^\dagger = \mathcal{O}$, with $\expval{\mathcal{O}}_{\rho(t)} = \Tr(\mathcal{O}\rho(t))$ i.e. the statistics of the measurements of observables satisfy the Born rule.
    \end{enumerate}
    Then if 
    \begin{equation}
        \label{Condition nonlinear extensions violate the MEP}
        \Tr(\mathcal{O}\rho(t))  \neq \sum_i p_i \Tr\Big(\mathcal{O} \psi_i(t)\Big) \, ,
    \end{equation}
    the MEP is violated \textit{statistically} for a system described by $\rho(t)$ through measurements of $\mathcal{O}$.
\end{lemma}

\begin{proof}
    Any proper mixed state at time $t_0$ can be written as
    \begin{equation}
        \rho_\text{proper}(t_0) = \sum_i p_i \ket{\psi_i(t_0)}\bra{\psi_i(t_0)} \in \mathscr{D}(\mathcal{H})
    \end{equation}
    given an ensemble of pure states $\{\ket{\psi_i}\}_i \in \mathcal{H}$ with associated probabilities $p_i, \, \sum_i p_i = 1$. Likewise, at time $t_0$, preparing the entangled state
    \begin{equation}
        \sum_i p_i \ket{i} \ket{\psi_i(t_0)} \in \tilde{\mathcal{H}} \otimes \mathcal{H}
    \end{equation}
    where the $\{\ket{i}\}_i$ are orthogonal states of an ancilla qudit, one can trace over $\tilde{\mathcal{H}}$ to get the improper mixed state
    \begin{equation}
        \rho_\text{improper}(t_0) = \sum_i p_i \ket{\psi_i(t_0)}\bra{\psi_i(t_0)} \in \mathscr{D}(\mathcal{H})
    \end{equation}
    i.e. $\rho_\text{proper}(t_0) = \rho_\text{improper}(t_0)$. If the Born rule is unchanged, i.e. $\forall t \in \mathcal{I} \subset \mathbb{R}, \, \expval{\mathcal{O}}_{\rho(t)} = \Tr(\mathcal{O}\rho(t))$, then:
    \begin{enumerate}
        \item in the proper case, by lemma \ref{Expval proper mixed states}, we must have that the statistics of the outcomes of any experiment on any observable $\mathcal{O}$ with $\rho_\text{proper}(t)$ must follow that of the statistics of outcomes with the pure states $\ket{\psi_i(t)}$ weighted by the $p_i$'s, i.e.
        \begin{align}\label{oproper}
            \expval{\mathcal{O}}_{\rho_\text{proper}(t)} &= \sum_i p_i \expval{\mathcal{O}}_{\psi_i(t)} \nonumber \\
            &= \sum_i p_i \Tr\Big(\mathcal{O} \psi_i(t)\Big) \, .
        \end{align}
        \item in the improper case, 
        \begin{equation}
            \rho_\text{improper}(t) = T_{\rho}(t,t_0)\Big[\sum_i p_i \psi_i(t_0)\Big] \equiv \rho(t) \, .
        \end{equation}
        If time-evolution is nonlinear, then there exists $\{ p_i \} , \{ \psi_i \}, t_0, t>t_0$ such that this is not equal to $\sum_i p_i \psi_i(t)$. Thus, since the statistics of the measurements of observables satisfy the Born rule,
        \begin{equation}\label{oimproper}
            \expval{\mathcal{O}}_{\rho_\text{improper}(t)} = \Tr(\mathcal{O} \rho(t))  \neq \sum_i p_i \Tr\Big(\mathcal{O} \psi_i(t)\Big)\, . 
        \end{equation}
       
    \end{enumerate}
From \eqref{Condition nonlinear extensions violate the MEP}, \eqref{oproper}, \eqref{oimproper} we have
\begin{equation}
    \expval{\mathcal{O}}_{\rho_\text{improper}(t)} \neq \expval{\mathcal{O}}_{\rho_\text{proper}(t)} \, , 
\end{equation}
i.e.~the MEP is violated for a system described by $\rho(t)$ through measurements of $\mathcal{O}$.
\end{proof}

\begin{theorem}
    \label{Nonlinear qm violates the MEP}
    Let 
    \begin{enumerate}
        \item $\{\ket{\psi_i(t_0)}\}_i \in \mathcal{H}$ be an ensemble of pure states with associated probabilities $\{p_i\}_i, \, \sum_i p_i = 1$, where $\mathcal{H}$ is the Hilbert space of the theory, 
        \item $\rho(t_0) = \sum_i p_i \psi_i(t_0) \in \mathscr{D}(\mathcal{H})$ where we wrote the pure state $\psi_i(t_0) = \ket{\psi_i(t_0)}\bra{\psi_i(t_0)} \in \mathscr{D}(\mathcal{H})$,
        \item $T_\rho(t,t_0) : \mathscr{D}(\mathcal{H}) \to \mathscr{D}(\mathcal{H})$ be a time-evolution operator for the state $\rho(t_0)$ such that $\rho(t) = T_\rho(t,t_0)[\rho(t_0)] \in \mathscr{D}(\mathcal{H})$, and likewise we write $\psi_i(t) = T_{\psi_i}(t,t_0)[\psi_i(t_0)]$.
    \end{enumerate}
    Then if
    \begin{equation}
        \rho(t) \neq \sum_i p_i \psi_i(t) \, ,
    \end{equation}
    the MEP is violated for a system described by $\rho(t)$.
\end{theorem}

\begin{proof}
    We write $\chi(t) := \sum_i p_i \psi_i(t)$. If $\rho(t) \neq \chi(t)$ then $\rho(t) - \chi(t) \neq 0$ so there exists a projection operator $P$ such that $\Tr[P\rho(t)] \neq \Tr[P \chi(t)]$ which probabilistically distinguishes both through measurements. Thus, from lemma \ref{Lemma Nonlinear qm violates the MEP}, the MEP is violated statistically through measurements of $P$.
\end{proof}

We extend this result to generalised probabilistic theories (GPTs) with nonlinear dynamics in Appendix \ref{GPTs with nonlinear dynamics violate the MEP}.

In particular, if the time-evolution of states is nonlinear, then the MEP will certainly be violated through some measurement, but may or may not be violated for a specific observable $\mathcal{O}$.  Furthermore, nonlinear time evolution is not sufficient for the violation of the MEP through measurements of \textit{arbitrary} operators: it may be that for some $\mathcal{O}$, $\Tr(\mathcal{O} \, \cdot)$ remains equal for proper and improper mixed states as is the case for the gravitational field in semiclassical gravity, which we look at below. 

Further note that nonlinear time evolution is not a necessary condition for MEP violation: modifications of the Born rule with linear dynamics can also violate the MEP (as we discuss in section \ref{Modifications of the Born rule violate the MEP section}).

\subsection{Semiclassical gravity and statistical GMEP violations}

Consider the case of semiclassical gravity and the thought-experiment of section \ref{Semiclassical gravity section}. If $\mathcal{O}$ is the gravitational field $\Phi$, we see that repeating the experiment a large number of times -- taking the time average and assuming that the matter from each experiment is ``cleared away" before the next -- yields no statistical (G)MEP violations as
\begin{multline}
    \expval{\Phi(\mathbf{y})}_\text{(proper)} = \abs{a_1}^2 . \frac{-Gm}{\abs{\mathbf{x}_1-\mathbf{y}}} + \abs{a_2}^2 . \frac{-Gm}{\abs{\mathbf{x}_2-\mathbf{y}}} \\ = \expval{\Phi(\mathbf{y})}_\text{(improper)} \, ,
\end{multline}
although, as noted above, a single experiment does violate the GMEP. Thus, statistical (G)MEP violations are inequivalent to one-shot (G)MEP violations for a given observable and experiment. 

However, since semiclassical gravity is a nonlinear theory, the arguments
of the previous section imply that it should produce statistical violations of the GMEP (as noted, though not with this terminology in \cite{tilloy_binding_2018}). 
We illustrate this with a simple example.
Consider now a mass interferometer built from a needle-shaped potential, as shown in Figure \ref{fig:mass interferometer}. 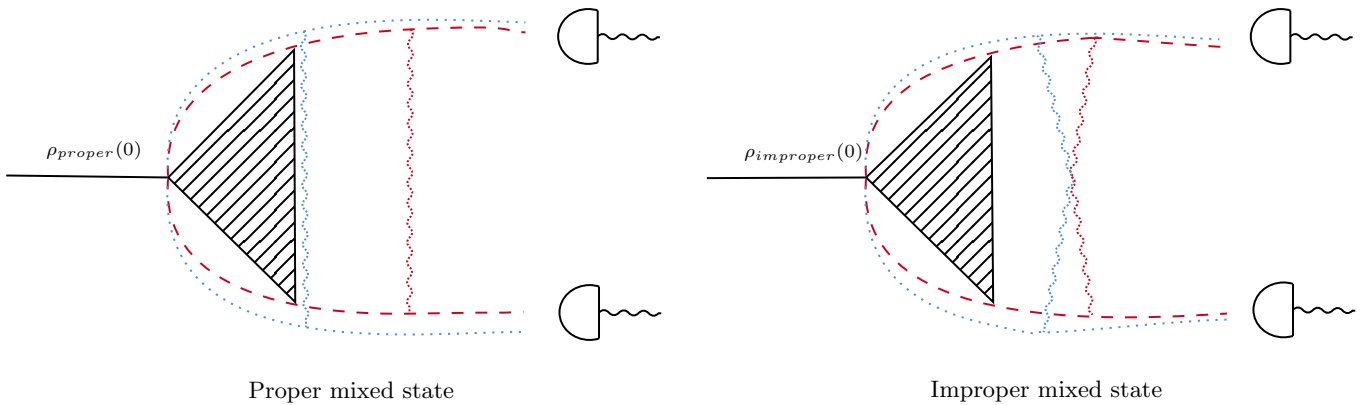
\begin{figure*}
    \centering
    \begin{tikzpicture}[x=0.75pt,y=0.75pt,yscale=-1,xscale=1]

%Straight Lines [id:da7437325315425201] 
\draw    (9.67,140) -- (91.33,141) ;
%Straight Lines [id:da17803350018503816] 
\draw    (363,141.33) -- (443.33,141) ;
%Curve Lines [id:da7186315227566991] 
\draw [color={rgb, 255:red, 208; green, 2; blue, 27 }  ,draw opacity=1 ] [dash pattern={on 4.5pt off 4.5pt}]  (91.33,141) .. controls (86.72,86.51) and (156.5,68.99) .. (213.5,66.25) .. controls (214.35,66.21) and (215.19,66.17) .. (216.02,66.13) .. controls (270.14,63.63) and (258.58,67.35) .. (271.5,67.75) ;
%Curve Lines [id:da06691021253013374] 
\draw [color={rgb, 255:red, 208; green, 2; blue, 27 }  ,draw opacity=1 ] [dash pattern={on 4.5pt off 4.5pt}]  (91.33,141) .. controls (86.09,198.17) and (156,210.75) .. (213,209.75) .. controls (270,208.75) and (255.66,209.7) .. (270.5,208.75) ;
%Curve Lines [id:da33083009086874293] 
\draw [color={rgb, 255:red, 208; green, 2; blue, 27 }  ,draw opacity=1 ] [dash pattern={on 4.5pt off 4.5pt}]  (443.33,141) .. controls (441.45,118.79) and (451.79,103.38) .. (468.56,92.74) .. controls (485.33,82.09) and (517.98,74.06) .. (546.58,71.46) .. controls (575.18,68.87) and (554.38,70.82) .. (558.23,70.58) .. controls (562.08,70.34) and (612.28,74.81) .. (626.5,75.25) ;
%Curve Lines [id:da19939058581645064] 
\draw [color={rgb, 255:red, 208; green, 2; blue, 27 }  ,draw opacity=1 ] [dash pattern={on 4.5pt off 4.5pt}]  (443.33,141) .. controls (438.62,192.38) and (493.59,207.46) .. (547.3,210.79) .. controls (577.21,212.65) and (606.73,210.86) .. (625.33,209.67) ;
%Shape: Arc [id:dp6676010849729526] 
\draw  [draw opacity=0] (307.95,82.97) .. controls (306.44,83.42) and (304.84,83.67) .. (303.17,83.67) .. controls (294.79,83.67) and (288,77.47) .. (288,69.83) .. controls (288,62.19) and (294.79,56) .. (303.17,56) .. controls (304.84,56) and (306.44,56.25) .. (307.95,56.7) -- (303.17,69.83) -- cycle ; \draw   (307.95,82.97) .. controls (306.44,83.42) and (304.84,83.67) .. (303.17,83.67) .. controls (294.79,83.67) and (288,77.47) .. (288,69.83) .. controls (288,62.19) and (294.79,56) .. (303.17,56) .. controls (304.84,56) and (306.44,56.25) .. (307.95,56.7) ;  
%Straight Lines [id:da615410687589983] 
\draw    (307.95,56.7) -- (307.95,82.97) ;
%Shape: Arc [id:dp762256020824281] 
\draw  [draw opacity=0] (657.28,82.97) .. controls (655.78,83.42) and (654.17,83.67) .. (652.5,83.67) .. controls (644.12,83.67) and (637.33,77.47) .. (637.33,69.83) .. controls (637.33,62.19) and (644.12,56) .. (652.5,56) .. controls (654.17,56) and (655.78,56.25) .. (657.28,56.7) -- (652.5,69.83) -- cycle ; \draw   (657.28,82.97) .. controls (655.78,83.42) and (654.17,83.67) .. (652.5,83.67) .. controls (644.12,83.67) and (637.33,77.47) .. (637.33,69.83) .. controls (637.33,62.19) and (644.12,56) .. (652.5,56) .. controls (654.17,56) and (655.78,56.25) .. (657.28,56.7) ;  
%Straight Lines [id:da9846031231107286] 
\draw    (657.28,56.7) -- (657.28,82.97) ;
%Shape: Arc [id:dp1088775568400655] 
\draw  [draw opacity=0] (308.61,222.3) .. controls (307.11,222.75) and (305.5,223) .. (303.83,223) .. controls (295.46,223) and (288.67,216.81) .. (288.67,209.17) .. controls (288.67,201.53) and (295.46,195.33) .. (303.83,195.33) .. controls (305.5,195.33) and (307.11,195.58) .. (308.61,196.03) -- (303.83,209.17) -- cycle ; \draw   (308.61,222.3) .. controls (307.11,222.75) and (305.5,223) .. (303.83,223) .. controls (295.46,223) and (288.67,216.81) .. (288.67,209.17) .. controls (288.67,201.53) and (295.46,195.33) .. (303.83,195.33) .. controls (305.5,195.33) and (307.11,195.58) .. (308.61,196.03) ;  
%Straight Lines [id:da881243351048034] 
\draw    (308.61,196.03) -- (308.61,222.3) ;
%Shape: Arc [id:dp8953886178751516] 
\draw  [draw opacity=0] (658.61,220.97) .. controls (657.11,221.42) and (655.5,221.67) .. (653.83,221.67) .. controls (645.46,221.67) and (638.67,215.47) .. (638.67,207.83) .. controls (638.67,200.19) and (645.46,194) .. (653.83,194) .. controls (655.5,194) and (657.11,194.25) .. (658.61,194.7) -- (653.83,207.83) -- cycle ; \draw   (658.61,220.97) .. controls (657.11,221.42) and (655.5,221.67) .. (653.83,221.67) .. controls (645.46,221.67) and (638.67,215.47) .. (638.67,207.83) .. controls (638.67,200.19) and (645.46,194) .. (653.83,194) .. controls (655.5,194) and (657.11,194.25) .. (658.61,194.7) ;  
%Straight Lines [id:da5939181580508643] 
\draw    (658.61,194.7) -- (658.61,220.97) ;
%Straight Lines [id:da9622131043007773] 
\draw    (307.95,69.83) .. controls (309.64,68.19) and (311.3,68.21) .. (312.95,69.9) .. controls (314.6,71.59) and (316.26,71.61) .. (317.95,69.97) .. controls (319.64,68.32) and (321.3,68.34) .. (322.95,70.03) .. controls (324.59,71.72) and (326.25,71.74) .. (327.94,70.1) .. controls (329.63,68.46) and (331.29,68.48) .. (332.94,70.17) .. controls (334.59,71.86) and (336.25,71.88) .. (337.94,70.23) -- (339.33,70.25) -- (339.33,70.25) ;
%Straight Lines [id:da2992034120256013] 
\draw    (308.61,209.17) .. controls (310.3,207.52) and (311.96,207.54) .. (313.61,209.23) .. controls (315.26,210.92) and (316.92,210.94) .. (318.61,209.3) .. controls (320.3,207.66) and (321.96,207.68) .. (323.61,209.37) .. controls (325.26,211.06) and (326.92,211.08) .. (328.61,209.43) .. controls (330.3,207.79) and (331.96,207.81) .. (333.61,209.5) .. controls (335.26,211.19) and (336.92,211.21) .. (338.61,209.56) -- (340,209.58) -- (340,209.58) ;
%Straight Lines [id:da49908682939618854] 
\draw    (657.28,69.83) .. controls (658.97,68.19) and (660.63,68.21) .. (662.28,69.9) .. controls (663.93,71.59) and (665.59,71.61) .. (667.28,69.97) .. controls (668.97,68.32) and (670.63,68.34) .. (672.28,70.03) .. controls (673.93,71.72) and (675.59,71.74) .. (677.28,70.1) .. controls (678.97,68.46) and (680.63,68.48) .. (682.28,70.17) .. controls (683.93,71.86) and (685.59,71.88) .. (687.28,70.23) -- (688.67,70.25) -- (688.67,70.25) ;
%Straight Lines [id:da8494599955930913] 
\draw    (658.61,207.83) .. controls (660.3,206.19) and (661.96,206.21) .. (663.61,207.9) .. controls (665.26,209.59) and (666.92,209.61) .. (668.61,207.97) .. controls (670.3,206.32) and (671.96,206.34) .. (673.61,208.03) .. controls (675.26,209.72) and (676.92,209.74) .. (678.61,208.1) .. controls (680.3,206.46) and (681.96,206.48) .. (683.61,208.17) .. controls (685.26,209.86) and (686.92,209.88) .. (688.61,208.23) -- (690,208.25) -- (690,208.25) ;
%Curve Lines [id:da2789900162490748] 
\draw [color={rgb, 255:red, 74; green, 144; blue, 226 }  ,draw opacity=1 ] [dash pattern={on 0.84pt off 2.51pt}]  (91.33,141) .. controls (86.59,190.48) and (120.1,209.64) .. (160.85,216.61) .. controls (200.25,223.34) and (246.43,218.67) .. (271.33,219) ;
%Curve Lines [id:da578786258685448] 
\draw [color={rgb, 255:red, 74; green, 144; blue, 226 }  ,draw opacity=1 ] [dash pattern={on 0.84pt off 2.51pt}]  (91.33,141) .. controls (87.44,96.56) and (120.01,76.26) .. (159.4,67.3) .. controls (175.89,63.54) and (194.38,62.02) .. (211.92,61.56) .. controls (236.26,60.92) and (258.77,62.35) .. (271.5,62.75) ;
%Curve Lines [id:da23077105373354456] 
\draw [color={rgb, 255:red, 74; green, 144; blue, 226 }  ,draw opacity=1 ] [dash pattern={on 0.84pt off 2.51pt}]  (443.33,141) .. controls (438.92,187) and (467.56,206.79) .. (504.35,214.94) .. controls (541.13,223.1) and (523.13,218.28) .. (532.89,219.09) .. controls (542.65,219.9) and (606.22,212.48) .. (627,212.75) ;
%Curve Lines [id:da9816054499000371] 
\draw [color={rgb, 255:red, 74; green, 144; blue, 226 }  ,draw opacity=1 ] [dash pattern={on 0.84pt off 2.51pt}]  (443.33,141) .. controls (438.81,89.33) and (482.76,73.37) .. (529.52,69.5) .. controls (565.42,66.52) and (602.97,70.67) .. (621.5,71.25) ;
%Straight Lines [id:da43185659444127245] 
\draw [color={rgb, 255:red, 74; green, 144; blue, 226 }  ,draw opacity=1 ] [dash pattern={on 0.75pt off 0.75pt}]  (159.4,67.3) .. controls (161.09,68.95) and (161.1,70.61) .. (159.45,72.3) .. controls (157.8,73.98) and (157.82,75.65) .. (159.5,77.3) .. controls (161.18,78.95) and (161.2,80.62) .. (159.55,82.3) .. controls (157.9,83.98) and (157.91,85.65) .. (159.59,87.3) .. controls (161.27,88.95) and (161.29,90.62) .. (159.64,92.3) .. controls (157.99,93.98) and (158.01,95.65) .. (159.69,97.3) .. controls (161.37,98.95) and (161.39,100.62) .. (159.74,102.3) .. controls (158.09,103.98) and (158.11,105.65) .. (159.79,107.3) .. controls (161.47,108.95) and (161.49,110.62) .. (159.84,112.3) .. controls (158.19,113.98) and (158.21,115.65) .. (159.89,117.3) .. controls (161.57,118.95) and (161.58,120.62) .. (159.93,122.3) .. controls (158.28,123.98) and (158.3,125.65) .. (159.98,127.3) .. controls (161.66,128.95) and (161.68,130.61) .. (160.03,132.29) .. controls (158.38,133.97) and (158.4,135.64) .. (160.08,137.29) .. controls (161.76,138.94) and (161.78,140.61) .. (160.13,142.29) .. controls (158.48,143.97) and (158.5,145.64) .. (160.18,147.29) .. controls (161.86,148.94) and (161.87,150.61) .. (160.22,152.29) .. controls (158.57,153.97) and (158.59,155.64) .. (160.27,157.29) .. controls (161.95,158.94) and (161.97,160.61) .. (160.32,162.29) .. controls (158.67,163.97) and (158.69,165.64) .. (160.37,167.29) .. controls (162.05,168.94) and (162.07,170.61) .. (160.42,172.29) .. controls (158.77,173.97) and (158.79,175.64) .. (160.47,177.29) .. controls (162.15,178.94) and (162.17,180.61) .. (160.52,182.29) .. controls (158.87,183.97) and (158.88,185.64) .. (160.56,187.29) .. controls (162.24,188.94) and (162.26,190.61) .. (160.61,192.29) .. controls (158.96,193.97) and (158.98,195.64) .. (160.66,197.29) .. controls (162.34,198.94) and (162.36,200.61) .. (160.71,202.29) .. controls (159.06,203.97) and (159.08,205.64) .. (160.76,207.29) .. controls (162.44,208.94) and (162.46,210.61) .. (160.81,212.29) -- (160.85,216.61) -- (160.85,216.61) ;
%Straight Lines [id:da6204081825399121] 
\draw [color={rgb, 255:red, 208; green, 2; blue, 27 }  ,draw opacity=1 ] [dash pattern={on 0.75pt off 0.75pt}]  (213.5,66.25) .. controls (215.16,67.92) and (215.15,69.59) .. (213.48,71.25) .. controls (211.81,72.92) and (211.81,74.58) .. (213.47,76.25) .. controls (215.13,77.92) and (215.12,79.59) .. (213.45,81.25) .. controls (211.78,82.91) and (211.77,84.58) .. (213.43,86.25) .. controls (215.09,87.92) and (215.08,89.59) .. (213.41,91.25) .. controls (211.74,92.92) and (211.74,94.58) .. (213.4,96.25) .. controls (215.06,97.92) and (215.05,99.59) .. (213.38,101.25) .. controls (211.71,102.91) and (211.7,104.58) .. (213.36,106.25) .. controls (215.02,107.92) and (215.01,109.59) .. (213.34,111.25) .. controls (211.67,112.92) and (211.67,114.58) .. (213.33,116.25) .. controls (214.99,117.92) and (214.98,119.59) .. (213.31,121.25) .. controls (211.64,122.91) and (211.63,124.58) .. (213.29,126.25) .. controls (214.95,127.92) and (214.94,129.59) .. (213.27,131.25) .. controls (211.6,132.92) and (211.6,134.58) .. (213.26,136.25) .. controls (214.92,137.92) and (214.91,139.59) .. (213.24,141.25) .. controls (211.57,142.91) and (211.56,144.58) .. (213.22,146.25) .. controls (214.88,147.92) and (214.87,149.59) .. (213.2,151.25) .. controls (211.53,152.92) and (211.53,154.58) .. (213.19,156.25) .. controls (214.85,157.92) and (214.84,159.59) .. (213.17,161.25) .. controls (211.5,162.91) and (211.49,164.58) .. (213.15,166.25) .. controls (214.81,167.92) and (214.8,169.59) .. (213.13,171.25) .. controls (211.46,172.92) and (211.46,174.58) .. (213.12,176.25) .. controls (214.78,177.92) and (214.77,179.59) .. (213.1,181.25) .. controls (211.43,182.91) and (211.42,184.58) .. (213.08,186.25) .. controls (214.74,187.92) and (214.73,189.59) .. (213.06,191.25) .. controls (211.39,192.92) and (211.39,194.58) .. (213.05,196.25) .. controls (214.71,197.92) and (214.7,199.59) .. (213.03,201.25) .. controls (211.36,202.91) and (211.35,204.58) .. (213.01,206.25) -- (213,209.75) -- (213,209.75) ;
%Straight Lines [id:da7052464805715728] 
\draw [color={rgb, 255:red, 74; green, 144; blue, 226 }  ,draw opacity=1 ] [dash pattern={on 0.75pt off 0.75pt}]  (547,140.75) .. controls (548.35,142.69) and (548.05,144.33) .. (546.11,145.67) .. controls (544.18,147.02) and (543.88,148.66) .. (545.23,150.59) .. controls (546.58,152.53) and (546.28,154.17) .. (544.34,155.51) .. controls (542.4,156.85) and (542.1,158.49) .. (543.45,160.43) .. controls (544.8,162.36) and (544.5,164) .. (542.57,165.35) .. controls (540.63,166.69) and (540.33,168.33) .. (541.68,170.27) .. controls (543.03,172.2) and (542.73,173.85) .. (540.8,175.2) .. controls (538.86,176.54) and (538.56,178.18) .. (539.91,180.12) .. controls (541.26,182.06) and (540.96,183.7) .. (539.02,185.04) .. controls (537.09,186.39) and (536.79,188.03) .. (538.14,189.96) .. controls (539.49,191.9) and (539.19,193.54) .. (537.25,194.88) .. controls (535.31,196.22) and (535.01,197.86) .. (536.36,199.8) .. controls (537.71,201.73) and (537.41,203.37) .. (535.48,204.72) .. controls (533.54,206.06) and (533.24,207.7) .. (534.59,209.64) .. controls (535.94,211.58) and (535.64,213.22) .. (533.7,214.56) -- (532.89,219.09) -- (532.89,219.09) ;
%Straight Lines [id:da8187454919120853] 
\draw [color={rgb, 255:red, 208; green, 2; blue, 27 }  ,draw opacity=1 ] [dash pattern={on 0.75pt off 0.75pt}]  (547,140.75) .. controls (548.89,142.16) and (549.14,143.81) .. (547.73,145.7) .. controls (546.33,147.59) and (546.58,149.24) .. (548.47,150.64) .. controls (550.36,152.05) and (550.61,153.7) .. (549.2,155.59) .. controls (547.79,157.48) and (548.04,159.13) .. (549.93,160.53) .. controls (551.82,161.94) and (552.07,163.59) .. (550.66,165.48) .. controls (549.26,167.37) and (549.51,169.02) .. (551.4,170.43) .. controls (553.29,171.83) and (553.54,173.48) .. (552.13,175.37) .. controls (550.72,177.26) and (550.97,178.91) .. (552.86,180.32) .. controls (554.75,181.72) and (555,183.37) .. (553.59,185.26) .. controls (552.19,187.15) and (552.44,188.8) .. (554.33,190.21) .. controls (556.22,191.62) and (556.47,193.27) .. (555.06,195.16) .. controls (553.65,197.05) and (553.9,198.7) .. (555.79,200.1) .. controls (557.68,201.51) and (557.93,203.16) .. (556.52,205.05) .. controls (555.12,206.94) and (555.37,208.59) .. (557.26,209.99) -- (557.3,210.29) -- (557.3,210.29) ;
%Straight Lines [id:da403512158658351] 
\draw [color={rgb, 255:red, 208; green, 2; blue, 27 }  ,draw opacity=1 ] [dash pattern={on 0.75pt off 0.75pt}]  (558.23,70.58) .. controls (559.61,72.49) and (559.35,74.13) .. (557.44,75.52) .. controls (555.53,76.9) and (555.27,78.54) .. (556.65,80.45) .. controls (558.03,82.36) and (557.77,84) .. (555.86,85.39) .. controls (553.95,86.78) and (553.69,88.42) .. (555.07,90.33) .. controls (556.45,92.24) and (556.19,93.88) .. (554.28,95.26) .. controls (552.37,96.65) and (552.11,98.29) .. (553.49,100.2) .. controls (554.87,102.11) and (554.61,103.75) .. (552.7,105.14) .. controls (550.79,106.53) and (550.53,108.17) .. (551.91,110.08) .. controls (553.29,111.99) and (553.03,113.63) .. (551.12,115.01) .. controls (549.21,116.4) and (548.95,118.04) .. (550.33,119.95) .. controls (551.71,121.86) and (551.45,123.5) .. (549.54,124.89) .. controls (547.63,126.27) and (547.37,127.91) .. (548.75,129.82) .. controls (550.13,131.73) and (549.87,133.37) .. (547.96,134.76) .. controls (546.05,136.15) and (545.79,137.79) .. (547.17,139.7) -- (547,140.75) -- (547,140.75) ;
%Straight Lines [id:da013073504315253226] 
\draw [color={rgb, 255:red, 74; green, 144; blue, 226 }  ,draw opacity=1 ] [dash pattern={on 0.75pt off 0.75pt}]  (529.52,69.5) .. controls (531.53,70.73) and (531.93,72.35) .. (530.71,74.36) .. controls (529.49,76.37) and (529.89,77.99) .. (531.9,79.21) .. controls (533.91,80.44) and (534.31,82.06) .. (533.09,84.07) .. controls (531.87,86.09) and (532.27,87.71) .. (534.29,88.92) .. controls (536.3,90.15) and (536.7,91.77) .. (535.48,93.78) .. controls (534.26,95.79) and (534.66,97.41) .. (536.67,98.64) .. controls (538.68,99.86) and (539.08,101.48) .. (537.86,103.49) .. controls (536.64,105.5) and (537.04,107.12) .. (539.05,108.35) .. controls (541.06,109.57) and (541.46,111.19) .. (540.24,113.2) .. controls (539.02,115.21) and (539.42,116.83) .. (541.43,118.06) .. controls (543.44,119.29) and (543.84,120.91) .. (542.62,122.92) .. controls (541.4,124.94) and (541.8,126.56) .. (543.82,127.77) .. controls (545.83,129) and (546.23,130.62) .. (545.01,132.63) .. controls (543.79,134.64) and (544.19,136.26) .. (546.2,137.48) -- (547,140.75) -- (547,140.75) ;
%Straight Lines [id:da709203636156301] 
\draw [pattern=_phbkm4h6k,pattern size=6pt,pattern thickness=0.75pt,pattern radius=0pt, pattern color={rgb, 255:red, 0; green, 0; blue, 0}]   (91.33,141) -- (155.4,203.9) -- (154.8,76.65) -- cycle ;
%Straight Lines [id:da06877182171626228] 
\draw [pattern=_7uoi8qg3w,pattern size=6pt,pattern thickness=0.75pt,pattern radius=0pt, pattern color={rgb, 255:red, 0; green, 0; blue, 0}]   (443.33,141) -- (507.4,203.9) -- (506.3,80.15) -- cycle ;

% Text Node
\draw (130.33,242.67) node [anchor=north west][inner sep=0.75pt]  [xscale=1,yscale=1] [align=left] {Proper mixed state};
% Text Node
\draw (474.33,242) node [anchor=north west][inner sep=0.75pt]  [xscale=1,yscale=1] [align=left] {Improper mixed state};
% Text Node
\draw (28.33,121.07) node [anchor=north west][inner sep=0.75pt]  [font=\scriptsize,xscale=1,yscale=1]  {$\rho _{proper}( 0)$};
% Text Node
\draw (380.33,122.4) node [anchor=north west][inner sep=0.75pt]  [font=\scriptsize,xscale=1,yscale=1]  {$\rho _{improper}( 0)$};
\end{tikzpicture}
    \caption{Mass interferometer with a needle-shaped potential (hashed area) splitting the wavefunctions into two branches. In the case of a proper mixed state, the gravitational interactions between the two branches (represented by curly lines) happen within a single pure state (e.g. blue or red) at every run of the experiment. In the case of an improper mixed state, the gravitational interactions between the two branches happen across the whole mixture.}
    \label{fig:mass interferometer}
\end{figure*} 
This mass interferometer could also be constructed using a beam splitter or a Stern-Gerlach apparatus, but internal degrees of freedom such as polarisation or spin are not required for the following argument.

We work in the nonrelativistic limit of semiclassical gravity \eqref{Semiclassical Field Equations}. In this regime, the gravitational field satisfies Poisson's equation \eqref{Semiclassical Newtonian potential}. Assuming a particle has mass $m$, the Schrödinger equation gets an extra term from the gravitational Hamiltonian and becomes the Schrödinger-Newton equation in the position basis \cite{Jones1995}
\begin{equation}
    \label{Schrodinger-Newton equation}
    i\hbar \frac{\partial \psi}{\partial t}(\mathbf{x},t)  = \Big(-\frac{\hbar^2}{2m}\nabla^2 + V(\mathbf{x}) + f(\abs{\psi}^2;\mathbf{x},t)\Big)\psi(\mathbf{x},t)
\end{equation}
where we assumed a time-independent potential $V(\mathbf{x})$, and
\begin{align}
    f(\abs{\psi}^2;\mathbf{x},t) &:= \int \hat{\rho}(\mathbf{x}) \Phi(\abs{\psi}^2;\mathbf{x}) d^3 \mathbf{x} \\
    &= - G m^2 \int \frac{\abs{\psi(\mathbf{y},t)}^2}{\abs{\mathbf{y}-\mathbf{x}}} d^3\mathbf{y}
\end{align}
is the nonlocal and nonlinear term introduced by semiclassical gravitational interactions. Consider two quantum states of the form
\begin{align}
    \ket{\psi_1(t)} &= a_{1,U} \ket{\psi_{1,U}(t)} + a_{1,D} \ket{\psi_{1,D}(t)} \\
    \ket{\psi_2(t)} &= a_{2,U} \ket{\psi_{2,U}(t)} + a_{2,D} \ket{\psi_{2,D}(t)}
\end{align}
where $a_{i,\alpha} \in \mathbb{C}, i = 1,2$ and $\alpha = U,D$, where $U$ and $D$ correspond to ``up" and ``down" branches of the wavefunctions, respectively, which are assumed to have orthogonal support. Given $p_1, p_2 \in [0,1]$ with $p_1+p_2 = 1$, we can then construct the proper mixed state at time $t=0$:
\begin{equation}
    \rho_\text{proper}(0) := p_1 \ket{\psi_1(0)}\bra{\psi_1(0)} + p_2 \ket{\psi_2(0)}\bra{\psi_2(0)}
\end{equation}
or, in the position basis,
\begin{align}
    \rho_\text{proper}(\mathbf{x},\mathbf{y};0) &:= \bra{\mathbf{x}}\rho_\text{proper}(0)\ket{\mathbf{y}} \\
    &= p_1 \psi_1(\mathbf{x},0)\psi_1^*(\mathbf{y},0) \nonumber \\ &+ p_2 \psi_2(\mathbf{x},0)\psi_2^*(\mathbf{y},0) \, .
\end{align}
Likewise, one can construct an improper mixed state at time $t=0$ by preparing an entangled state
\begin{equation}
    \sqrt{p_1} \ket{0} \ket{\psi_1(0)} + \sqrt{p_2} \ket{1} \ket{\psi_2(0)}
\end{equation}
where $\ket{0},\ket{1}$ are orthogonal states of an ancilla qubit, and tracing out the ancilla Hilbert space to get
\begin{equation}
    \rho_\text{improper}(0) := p_1 \ket{\psi_1(0)}\bra{\psi_1(0)} + p_2 \ket{\psi_2(0)}\bra{\psi_2(0)}
\end{equation}
with $\rho_\text{improper}(0) = \rho_\text{proper}(0)$, and likewise in the position basis $\rho_\text{improper}(\mathbf{x},\mathbf{y};0) = \rho_\text{proper}(\mathbf{x},\mathbf{y};0)$.

The time evolution for the proper mixed state will follow those of the pure states \eqref{Schrodinger-Newton equation} weighted by the probabilities, so that
\begin{equation}
    \rho_\text{proper}(\mathbf{x},\mathbf{y};t) = p_1 \psi_1(\mathbf{x},t)\psi_1^*(\mathbf{y},t) + p_2 \psi_2(\mathbf{x},t)\psi_2^*(\mathbf{y},t) 
\end{equation}
such that the statistics of position measurements at spacetime position $(\mathbf{x},t)$ follow, by lemma \ref{Expval proper mixed states}, 
\begin{align}
    \hat{\expval{X}}_{\rho_\text{proper}(\mathbf{x},\mathbf{x};t)} &= p_1 \hat{\expval{X}}_{\psi_1(\mathbf{x},t)} + p_2 \hat{\expval{X}}_{\psi_2(\mathbf{x},t)} \, .
\end{align}
Once again, we stress that this just follows the weighted statistics of position measurements on pure states. On the other hand, in the improper case, the nonlinear term $f(\rho_\text{improper};\mathbf{x},t)$ will typically couple the time evolution the two branches of the mixture with
\begin{equation}
    f(\rho_\text{improper};\mathbf{x},t) = - G m^2 \int \frac{\rho_\text{improper}(\mathbf{y},\mathbf{y};t)}{\abs{\mathbf{y}-\mathbf{x}}} d^3\mathbf{y} \, , 
\end{equation}
where in particular
\begin{multline}
    f(\rho_\text{improper};\mathbf{x},0) \\ = - G m^2 \int \frac{p_1\abs{\psi_{1}(\mathbf{y},0)}^2+ p_2 \abs{\psi_{2}(\mathbf{y},0)}^2}{\abs{\mathbf{y}-\mathbf{x}}} d^3\mathbf{y}
\end{multline}
rather than just depending on either $\psi_1(\mathbf{x},0)$ or $\psi_2(\mathbf{x},0)$ respectively, as is the case for proper mixed states. Thus, in general,
\begin{align}
    \rho_\text{improper}(\mathbf{x},\mathbf{y};t) &\neq p_1 \psi_1(\mathbf{x},t)\psi_1^*(\mathbf{y},t) + p_2 \psi_2(\mathbf{x},t)\psi_2^*(\mathbf{y},t) 
\end{align}
for generic initial states at $t > 0$. Thus 
\begin{equation}
    \hat{\expval{X}}_{\rho_\text{improper}(\mathbf{x},\mathbf{x};t)} \neq p_1 \hat{\expval{X}}_{\psi_1(\mathbf{x},t)} + p_2 \hat{\expval{X}}_{\psi_2(\mathbf{x},t)}
\end{equation}
i.e. 
\begin{equation}
\hat{\expval{X}}_{\rho_\text{improper}(\mathbf{x},\mathbf{x};t)} \neq \hat{\expval{X}}_{\rho_\text{proper}(\mathbf{x},\mathbf{x};t)} \, . 
\end{equation}
Hence, semiclassical gravity violates the GMEP both one-shot -- with measurements of the gravitational field in the experiment of section \ref{Semiclassical gravity section} -- and statistically -- with position measurements in a mass interferometer. 

We leave as an open question whether semiclassical gravity necessarily statistically violates the GWMEP.

\subsection{Modifications of the Born rule violate the MEP}

\label{Modifications of the Born rule violate the MEP section}

Consequences of modifying the Born rule have previously been
studied by Aaronson \cite{Aaronson2004}, who noted that $p$-norm
generalisations imply (inter alia) the distinguishability of 
non-orthogonal states, and 
by Galley and Masanes \cite{Galley2018}, who showed
that modifying the Born rule violates the PP.   
Their result implies the weaker result of this section, 
which we nonetheless include since it is simple to 
state and show in our restricted context, and makes more
complete our discussion of MEP violations in this context.  

We write $\expval{\expval{ O }}$ for the expectation value
of an observable $O$ defined via a modified Born rule. Let $\{\ket{\psi_i}\}_i \in \mathcal{H}$ be an ensemble of pure states with associated probabilities $\{p_i\}_i, \, \sum_i p_i = 1$ and (pure) density operator $\psi_i = \ket{\psi_i}\bra{\psi_i} \in \mathscr{D}(\mathcal{H})$, where $\mathcal{H}$ is the Hilbert space of the theory, and $\mathcal{O} \in \mathcal{L}(\mathcal{H})$ be an observable, i.e. $\mathcal{O}^\dagger = \mathcal{O}$. Let $\rho_\text{proper} = \sum_i p_i \psi_i$. As previously argued in lemma \ref{Expval proper mixed states},
\begin{equation}
    \expval{\expval{\mathcal{O}}}_{\rho_\text{proper}} = \sum_i p_i \expval{\expval{\mathcal{O}}}_{\psi_i} \, .
\end{equation}
Given an entangled state
\begin{equation}
    \sum_i p_i \ket{i} \ket{\psi_i}
\end{equation}
where the $\ket{i}$ are orthogonal ancilla qudits, we can trace over these to get $\rho_\text{improper} = \sum_i p_i \psi_i = \rho_\text{proper}$. In that case,
\begin{equation}
    \expval{\expval{\mathcal{O}}}_{\rho_\text{improper}} = \expval{\expval{\mathcal{O}}}_{\sum_i p_i \psi_i}
\end{equation}
so if 
\begin{equation}
    \expval{\expval{\mathcal{O}}}_{\sum_i p_i \psi_i} \neq \sum_i p_i \expval{\expval{\mathcal{O}}}_{\psi_i}
\end{equation}
then $\expval{\expval{\mathcal{O}}}_{\rho_\text{improper}} \neq \expval{\expval{\mathcal{O}}}_{\rho_\text{proper}}$, i.e. nonlinear modifications to the Born rule violate the MEP statistically for a system described by $\rho$ through measurements of $\mathcal{O}$\footnote{One then expects de Broglie-Bohm theory to violate the MEP in quantum non-equilibrium \cite{Valentini1992,Valentini2005}. Indeed, initial conditions give different probability distributions for proper and improper mixed states. Whether these are experimentally measurable may depend on the initial state, the details of the dynamics and how fast one reaches equilibrium.}. 

Note that even affine modifications of the Born rule, of the form
\begin{equation}
    \expval{\expval{\mathcal{O}}}_\rho = \frac{k\Tr(\mathcal{O}\rho)+c}{\sum_{i=1}^N (k\Tr(\Pi_i \rho) + c)} \, ,
\end{equation}
violate the MEP, for $N$ measurement outcomes with projectors $\Pi_i$ and constants $k,c \in \mathbb{R}^*$ where the denominator provides probability normalisation. Indeed, for $\rho = \sum_j p_j \psi_j$ where the $\psi_j$ are pure, 
\begin{multline}
    \expval{\expval{\mathcal{O}}}_{\sum_j p_j \psi_j} = \frac{k\sum_j p_j \Tr(\mathcal{O}\psi_j)+c}{\sum_{i=1}^N (k \sum_j p_j \Tr(\Pi_i \psi_j) + c)} \\
    \neq \sum_j p_j \frac{k\Tr(\mathcal{O}\psi_j)+c}{\sum_{i=1}^N (k\Tr(\Pi_i \psi_j) + c)}
\end{multline}
in general. Nonlinear modifications of the Born rule thus generically violate the MEP - this result is independent of the dynamics, which may or may not be linear.

Note that the converse is not necessarily true, as we have shown in the case of semiclassical gravity which still follows the Born rule. 
We emphasize again that theorem \ref{Nonlinear qm violates the MEP} as well as the result above do not imply the equivalence between violations of the MEP and signalling issues arising from nonlinearities in the dynamics or in the computation of probabilities. There exist nonlinear extensions of the dynamics \cite{Kent2005} and of the Born rule \cite{Helou2017,Galley2018} that do not allow superluminal signalling.

\section{Black Holes and Hawking Radiation}

Semiclassical gravity can also explicitly violate the GMEP by distinguishing proper and improper mixtures of energy eigenstates. This has significant implications for the treatment of thermal ensembles. On the one hand, we may treat canonical ensembles as statistical, taking the thermal state of a system to be in one of the energy levels $\{E_1,E_2,...,E_N\}$ with some associated probability $\{p_1,p_2,...,p_N\}=\{\exp(-\frac{E_1}{k_\text{B} T}),\exp(-\frac{E_2}{k_\text{B} T}),...,\exp(-\frac{E_N}{k_\text{B} T})\}$ originating from classical uncertainty, so that the resulting Gibbs state is a proper mixed state
\begin{align}
    \rho_G^{\text{proper}} &= \frac{1}{Z} \sum_{k=1}^N p_k \ket{E_k}\bra{E_k} \\ &= \frac{1}{Z} \sum_{k=1}^N \exp(-\frac{E_k}{k_\text{B} T}) \ket{E_k}\bra{E_k} \, , \label{proper Gibbs}
\end{align}
where $Z = \sum_k \exp(-\frac{E_k}{k_\text{B} T})$ is the partition function serving as probability normalisation $\Tr(\rho)=1$. 

On the other hand, we may also derive the Gibbs state by enlarging the Hilbert space $\mathcal{H}_1$ to $\mathcal{H}_1\otimes\mathcal{H}_2$, 
where $\dim( \mathcal{H}_2 ) \geq \dim ( \mathcal{H}_1 )$, considering the so-called thermofield double state \cite{Israel1976}
\begin{equation}
    \ket{\psi} = \frac{1}{\sqrt{Z}} \sum_{k=1}^N \exp(-\frac{E_k}{2 k_\text{B} T}) \ket{E_k}_1 \otimes \ket{E_k}_2 \, , 
\end{equation}
and taking the partial trace with respect to $\mathcal{H}_2$, which yields
\begin{align}
    \rho_G^{\text{improper}} &= \Tr_{\mathcal{H}_2}(\ket{\psi}\bra{\psi}) \\ 
    &= \frac{1}{Z} \sum_{k=1}^N \exp(-\frac{E_k}{k_\text{B} T})  \ket{E_k}\bra{E_k} \, . \label{improper Gibbs}
\end{align}
We see that equations \eqref{proper Gibbs} and \eqref{improper Gibbs} are equal: they both describe a Gibbs state, but are ontologically different, describing respectively a proper and an improper mixed state. 

In conventional quantum theory these are not experimentally distinguishable, but in the context of semiclassical gravity, they generally are.
As in the thought experiment of Figure \ref{fig:MEP violation}, 
semiclassical gravity gives a gravitational field corresponding to  a single energy eigenstate $\ket{E_k}\bra{E_k}$ for some fixed $k$ from the proper mixed state \eqref{proper Gibbs}, but a gravitational field determined by the full improper mixed state \eqref{improper Gibbs} from the thermofield double state. These are experimentally distinguishable, and this is relevant in the context of black hole physics. 

\subsection{Hawking radiation is improper}

We consider a Schwarzschild spacetime. Let the black hole region be $\mathcal{B} = \{(t,r,\theta,\phi) | 0 < r < 2GM/c^2\}$ and the black hole exterior $\mathcal{M} = \{(t,r,\theta,\phi) | r > 2GM/c^2\}$. In practice, the Hawking spectrum \eqref{Hawking spectrum} can be derived in many different ways \cite{Hawking1975,Parker1975,Wald1975,Hartle1976,Israel1976,Unruh1976,Wald1984}, though we here start by outlining two (Lorentzian) pictures which can equivalently be adopted for the following discussion:
\begin{enumerate}
    \item That of Werner Israel \cite{Israel1976}, depicted in Figure \ref{fig:Penrose}. One may formally analytically extend the spacetime through the Kruskal extension and augment the physical Fock space $\mathcal{F}$ of a hypersurface $\Sigma \subset \mathcal{M}$ to $\mathcal{F}\otimes\mathcal{\tilde{F}}$. Here, $\mathcal{\tilde{F}}$ corresponds to the Fock space of a hypersurface in the hashed region in Figure \ref{fig:Penrose} i.e. $\tilde{\Sigma} \subset \bar{\mathcal{M}}$ where $\bar{\mathcal{M}}$ is the dual region to $\mathcal{M}$ of the Kruskal extension. The thermofield double state on $\Sigma \cup \tilde{\Sigma}$ is then \cite{Israel1976}
\begin{equation}
    \label{Thermofield double state}
    \ket{\psi} = \frac{1}{\sqrt{Z}} \sum_n  \exp(-\frac{E_n}{2k_\text{B} T_\text{H}}) \ket{n}_{\mathcal{F}} \ket{n}_{\mathcal{\tilde{F}}} \,.
\end{equation} This is entangled. The (improper) reduced density matrix of the state in $\Sigma \subset \mathcal{M}$ is just its partial trace with respect to $\mathcal{\tilde{F}}$, which is given by summing over the unknown states in the hashed region, yielding the Gibbs state
\begin{align}
    \rho_{\Sigma} &= \Tr_{\tilde{\mathcal{F}}}(\ket{\psi}\bra{\psi}) \label{density matrix spacetime} \\ &= \frac{1}{Z} \sum_n \exp(-\frac{E_n}{k_\text{B} T_\text{H}}) \ket{n}\bra{n}_\mathcal{F} \label{density matrix proper Hawking} \\
    &= \frac{1}{Z} \exp(-\frac{\hat{H}_{\text{Hawking}}}{k_\text{B} T_\text{H}}) \, , 
\end{align}
where $Z = \Tr(\exp(-\frac{\hat{H}_{\text{Hawking}}}{k_\text{B} T_\text{H}})) = \sum_n \exp(-\frac{E_n}{k_\text{B} T_\text{H}})$ is the partition function serving as probability normalisation $\Tr(\rho)=1$, and 
\begin{equation}
    \label{Hawking Hamiltonian}
    \hat{H}_{\text{Hawking}} = \int_{0}^{+\infty} \hbar \omega \hat{b}^\dagger(\omega) \hat{b}(\omega) d\omega
\end{equation}
is the Hawking Hamiltonian with Hawking temperature $T_\text{H} = \frac{\hbar c^3}{8\pi G M k_\text{B}}$. From this Hamiltonian, one can derive the blackbody spectrum from the expected number of late time ``out" particles with frequency $\omega$
\begin{equation}
    \label{Hawking spectrum}
    \expval{\hat{b}^\dagger(\omega) \hat{b}(\omega)} = \frac{\Gamma(\omega)}{e^{\hbar \omega/k_\text{B} T_\text{H}}-1} \, , 
\end{equation}
which indeed corresponds to Hawking radiation.
    \item That of Hawking and Wald \cite{Hawking1975,Wald1975}, which is depicted in Figure \ref{collapsing}. One here works with a collapsing black hole spacetime, in which case the thermofield double lies on a hypersurface $\Xi$ with two connected components: one at $\mathcal{I}^+$ (on which parochial observers measure what escapes from the black hole) and the other beyond the horizon (on which parochial observers measure what fell into the black hole). Because one does not observe inside the hole, one has to sum over all possibilities for the surface inside the hole and so obtains a density matrix describing a mixed state \cite{Hawking1982b,Hawking1987}. Thus, the resulting (improper) reduced density operator on $\mathcal{F}_{out}$, obtained after tracing out $\mathcal{F}_{int}$ (assuming $\mathcal{F}_{in} \cong \mathcal{F}_{out} \otimes \mathcal{F}_{int}$ where $\mathcal{F}_{in}$ is the Fock space on $\mathcal{I}^-$, $\mathcal{F}_{out}$ that on $\mathcal{I}^+$ and $\mathcal{F}_{int}$ that in $\Xi \cap \mathcal{B}$), is again a Hawking-Gibbs state of the form \eqref{density matrix proper Hawking} obtained by tracing out a thermofield double of the form \eqref{Thermofield double state}.
\end{enumerate}  

\begin{figure}[t!]
    \centering
    \subfloat[\justifying Penrose diagram of a maximally extended Schwarzschild black hole. The hashed region corresponds to regions of the extended spacetime for which we have no information: it is traced over. \label{fig:Penrose}]{
}
    \hfill
    \subfloat[\justifying Penrose diagram of a collapsing black hole. The hashed region $\mathcal{B}$ is being traced out, yielding an improper Hawking-Gibbs state on $\mathcal{I}^+$.\label{collapsing}]{%
}
    \caption{\justifying Penrose diagrams of black hole spacetimes highlighting two different ways to recover Hawking radiation. }
\end{figure}

Now, as Wald puts it \cite{Wald1975}, ``the density matrix for emission of particles to infinity at late times by spontaneous particle creation resulting from spherical gravitational collapse to a black hole is identical in all aspects to that of black body thermal emission at temperature $k_B T_H = \frac{\hbar \kappa}{2\pi}$", where $\kappa$ is the surface gravity.   One might take this to indicate that one's perspective on the nature of this Gibbs state is purely interpretational. However, in the context of M{\o}ller-Rosenfeld semiclassical gravity, different perspectives on the derivation of Hawking radiation give different experimental predictions. Whilst in a unitary theory this interpretational question is operationally irrelevant, it becomes of importance in the context of nonlinear modifications to quantum theory such as M{\o}ller-Rosenfeld semiclassical gravity. In the improper case, the semiclassical gravitational field backreacts from a weighted average of the energy levels; in the proper case, it backreacts from only one of them.

In versions of quantum theory in which unitary evolution is universal and the initial state is pure, all thermal states are improper mixtures. Let us underline that the most commonly accepted derivations of Hawking radiation -- the real spacetime pictures and the Euclidean picture -- do indeed lead to an improper Hawking-Gibbs state.

The real (Lorentzian) spacetime pictures \cite{Israel1976,Wald1984}, described above, obviously lead to an improper mixture: the \enquote{inaccessible} or black hole Fock spaces are being traced out. The Gibbs state \eqref{density matrix spacetime} is an improper mixture arising from restricting our consideration of the whole system -- the mixed state on $\Sigma$ (or $\mathcal{I}^+$, respectively) can be purified to the thermofield double state \eqref{Thermofield double state} living on the whole of $\Sigma \cup \tilde{\Sigma}$ (respectively, the whole of $\Xi$).

Another commonly considered approach is the Euclidean derivation of Hawking radiation. Here, one Wick rotates the Schwarzschild exterior solution and then deduces the (Hawking) temperature measured by an observer that is more than a few Schwarzschild radii away from the black hole \cite{Almheiri2021} à la Unruh \cite{Unruh1976} and Tolman \cite{Tolman1930} - the derivation is recalled in Appendix \ref{app:Unruh-Tolman}. The resulting vacuum state is an improper thermal mixture. Indeed, one does restrict the observer's consideration to the Rindler wedge associated to the observer's uniformly accelerating trajectory (needed to stay at a fixed radius and avoid falling into the black hole). This restriction is precisely what makes the mixture improper -- restricting the algebra of observables to that wedge makes the pure vacuum \enquote{become} an improper KMS state \cite{Bisognano1975,Bisognano1976,Earman2011} -- an analogue of an improper Gibbs state beyond type I von Neumann algebras.

Another example of the improper nature of this mixture can be seen through the use of Euclidean path integrals \cite{Hartman2015} to compute Hawking radiation - the derivation is recalled in Appendix \ref{app:Euclidean path integral}. In this case one extends the spacetime by sending the imaginary time $t_E$ to $t_E + \pi$ to land in $\bar{\mathcal{M}}$. The path integral on this Euclidean black hole geometry then yields an entangled state of the form \eqref{Thermofield double state} on hypersurfaces living in the extended spacetime $\bar{\mathcal{M}} \cup \mathcal{M}$.  Thus, restricting consideration to $\mathcal{I}^+$ in $\mathcal{M}$, one recovers the (improper) Hartle-Hawking state \eqref{density matrix spacetime}.

Hence, such derivations are still mutually consistent in the context of a nonlinear theory such as M{\o}ller-Rosenfeld semiclassical gravity. Any derivation that leads to a description of Hawking radiation as a proper mixture would be inconsistent with either, however. For example, a \enquote{statistical ensemble} understanding of Hawking radiation, which may seem operationally intuitive, would be inconsistent with the above derivations if we assume that M{\o}ller-Rosenfeld semiclassical gravity holds. We stress that, when working in nonlinear extensions of quantum theory, one must carefully consider situations that concern thermal solutions since proper (statistical) and improper thermal states are now inequivalent.

\subsection{M{\o}ller-Rosenfeld semiclassical gravity is not the semiclassical limit of quantum gravity}

It is often (e.g. see \cite{Wald1984,Kuo1993,Lowe2023}) argued that a semiclassical approximation to quantum gravity, taken to be the replacement of $T_{\mu\nu} \to \expval{\hat{T}_{\mu\nu}}$ in the classical Einstein field equations (thus leading to a limiting theory of the form of M{\o}ller-Rosenfeld semiclassical gravity), is appropriate when quantum fluctuations in the matter fields are suppressed but at the same time still overwhelm the fluctuations in the metric. The quantum backreaction effects of gravitons in the presence of a single matter field can be
comparable in magnitude to those of the matter field.  However, when $N$ matter fields are present, the quantum backreaction effects from the gravitons are $\mathcal{O}(1/N)$ compared to those of the matter.
One can then recover a semiclassical approximation to quantum gravity
by taking a large $N$ expansion (see Appendix \ref{app:Gravitational backreactions in the large N limit}).

It is however worth noting that, even though the form of the dynamical equations are the same, M{\o}ller-Rosenfeld semiclassical gravity and the semiclassical limit of quantum gravity make differing predictions in the same context of Hawking radiation if we make the 
(strong and non-obvious, though maybe common) assumption that there is a consistent treatment of measurements of the gravitational field within the latter. Even though the resulting Hawking-Gibbs state is improper in both cases, a measurement of the gravitational field would identify a particular sub-component of the state in the latter case.
The gravitational backreaction would then follow that of a proper mixture in the semiclassical limit of quantum gravity, as was highlighted throughout section \ref{Proper and improper mixed states}. As we have shown in this paper, this is different to the gravitational backreaction of an improper mixture in M{\o}ller-Rosenfeld semiclassical gravity. That is, an observer at $\mathcal{I}^+$ performing a Cavendish experiment will see different outcomes in the two theories: a one-shot experiment would give information about whether the observer lives in a quantum gravitational world or a M{\o}ller-Rosenfeld semiclassical world.

The quantum state used in the semiclassical Einstein field equations in the semiclassical limit of quantum gravity should thus always be interpreted as a \emph{proper mixture} to be consistent with the  behaviour expected from a unitary quantum gravity theory combined with measurements of the gravitational field. Importantly, this difference in behaviour between M{\o}ller-Rosenfeld semiclassical gravity and the semiclassical limit of quantum gravity is independent of the number of matter fields present, and should be understood as coming from the fact that the semiclassical limit of quantum gravity comes from a theory which is intrinsically unitary, while M{\o}ller-Rosenfeld semiclassical gravity is not. In other words, M{\o}ller-Rosenfeld semiclassical gravity is not the semiclassical limit of quantum gravity in the context of black hole spacetimes.\footnote{See \cite{Terno2024} for further discussion of the semiclassical limit of quantum gravity, with some comments on the present paper.}

\section{Acknowledgements}
We would like to thank Sean Hartnoll, Don Page and Aron Wall for interesting discussions related to the nature of the Hartle-Hawking state.
We acknowledge financial support from the
UK Quantum Communications Hub grant no. 
EP/T001011/1 and UK-Canada Quantum for Science research collaboration grant OPP640. 
S.F. is funded by a studentship from the Engineering and Physical Sciences Research Council. A.K. was supported in part by Perimeter Institute for
Theoretical Physics. Research at Perimeter Institute is supported by
the Government of Canada through the Department of Innovation, Science
and Economic Development and by the Province
of Ontario through the Ministry of Research, Innovation and Science. 

\appendix

\section{General Probabilistic Theories with nonlinear dynamics violate the MEP statistically}

\label{GPTs with nonlinear dynamics violate the MEP}

We extend the discussion of section \ref{Nonlinear QM violates the MEP statistically} to cover MEP violations of GPTs with nonlinear dynamics, which we now review \cite{Barnum2007}. In GPTs, a physical system is characterised by a state space $\Omega$ which is assumed to be a convex vector space. An element $\omega \in \Omega$ is called a state, and, by convexity, the probabilistic mixture $p_1 \omega_1 + p_2 \omega_2 \in \Omega$ is a state for any $\omega_1,\omega_2 \in \Omega$ and $p_1 + p_2 = 1$. A state is called \textit{pure} if it cannot be written as a convex combination of other states (i.e. if it is an extreme point of $\Omega$), and \textit{mixed} otherwise. 

The space of linear functionals $f: \Omega \to \mathbb{R}$ is denoted $A(\Omega)$, which is ordered in the sense that $f \leq g \Leftrightarrow f(\omega) \leq g(\omega) \, \forall \omega \in \Omega$. The zero and unit functionals $0_{A(\Omega)}$ and $1_{A(\Omega)}$, respectively, are defined so that $\forall \omega \in \Omega$, $0_{A(\Omega)}(\omega) = 0$ and $1_{A(\Omega)}(\omega) = 1$. An \textit{effect}, denoted $e_j$, is an element of $[0_{A(\Omega)},1_{A(\Omega)}]$, in the sense that it is a linear functional $e_j : \Omega \to [0,1]$. Effects are interpreted as events associated with the system considered, with occurrence probability $e_j(\omega)$ when the system is in state $\omega$. 

A (discrete) \textit{observable} is a mapping from a finite set $E$ to $A(\Omega)$:
\begin{align}
    e : E &\to A(\Omega) \\
    i &\mapsto e_i
\end{align}
satisfying $e_i \geq 0_{A(\Omega)} \, \forall i \in E$ and $\sum_{i\in E} e_i = 1_{A(\Omega)}$. Then, $\forall \omega \in \Omega$, $e_i(\omega) = p(i)$ where $p \in \Delta(E)$ is a probability weight (living in the set of all classical probability distributions over $E$, $\Delta(E)$), such that each $e_i$ is an effect. Here, $i \in E$ is then understood as a measurement outcome. The \textit{measurement} of an observable $e$ on a state $\omega$, denoted $e(\omega)$, is then defined to be the set of effects associated with that observable that sum to unit probability: 
\begin{equation}
    e(\omega) = \Big\{e_i | \sum_{i \in E} e_i(\omega) = 1\Big\} \, .
\end{equation}

An \textit{operation} is a linear mapping $\kappa : \Omega \to \Omega'$. Here, we will be interested in the case where $\Omega' = \Omega$, i.e. where $\kappa \in \mathcal{L}(\Omega)$. There is an associated dual linear transformation $\kappa^* : A(\Omega) \to A(\Omega)$ defined as $\kappa^*(e_i)(\omega) = e_i(\kappa(\omega))$ for all effects $e_i \in [0_{A(\Omega)},1_{A(\Omega)}]$ and states $\omega \in \Omega$. 

In standard quantum theory, $\omega \equiv \rho \in \mathscr{D}(\mathcal{H})$, $e_i \equiv \Tr(\Pi_i \, \cdot)$ for projective measurements where $\Pi_i$ is the projection operator associated with outcome $i$ of some observable (now seen as an operation) $\kappa \equiv \mathcal{O} \in \mathcal{L}(\mathcal{H})$, $e_i(\omega) = p(i) \equiv \Tr(\Pi_i \omega)$, and $e(\omega)$ is the set of all such $\Tr(\Pi_i \, \cdot)$ associated with that measurement. Then
\begin{equation}
    \expval{\mathcal{O}}_\rho = \sum_{i \in E} \Tr(\Pi_i \mathcal{O} \rho) \equiv \sum_{i \in E} e_i(\kappa(\omega))
\end{equation}
as $\sum_{i \in E} \Pi_i = 1$.

We now provide a generalisation of lemma \ref{Expval proper mixed states} to GPTs, whose statement is again both trivial yet important to highlight.
\begin{lemma}
    \label{Expval proper mixed states GPTs}
    Let
    \begin{enumerate}
        \item $\{\omega_i\}_i \in \Omega$ be an ensemble of pure states with associated probabilities $\{p_i\}_i, \, \sum_i p_i = 1$, where $\Omega$ is the convex state space of the GPT, 
        \item $\rho = \sum_i p_i \omega_i \in \Omega$ be a mixed state,
        \item $\kappa \in \mathcal{L}(\Omega)$ be an operation on $\Omega$,
        \item $e = \{e_j | \sum_{j \in E} e_j(\rho) = 1\}$ be a measurement on $\rho$, where $e_j : \Omega \to [0,1]$ are effects and $E$ is the set of possible measurement outcomes.
    \end{enumerate}
    Then
    \begin{equation}
        \sum_{j\in E} e_j(\kappa(\rho)) = \sum_i p_i \sum_{j\in E} e_j(\kappa(\omega_i)) \, .
    \end{equation}
\end{lemma}

\begin{proof}
    This follows by the linearity of operations and effects:
    \begin{align}
        \sum_{j\in E} e_j(\kappa(\rho)) &= \sum_{j\in E} e_j(\kappa(\sum_i p_i \omega_i)) \\ &= \sum_{j\in E} e_j(\sum_i p_i \kappa(\omega_i)) \\ &= \sum_{j\in E} \sum_i p_i e_j(\kappa(\omega_i))
    \end{align}
    and the result follows (since $E$ is taken to be finite).
\end{proof}

We now generalise the discussion of the MEP to cover GPTs.
\begin{quote}
  {\bf MEP:}  \qquad   Let $\{\omega_i\}_i \in \Omega$ be an ensemble of pure states with associated probabilities $\{p_i\}_i, \, \sum_i p_i = 1$, where $\Omega$ is the state space representing a system $S$, and $\rho_{\text{proper}}$ be the corresponding proper mixed state.   Let $\omega \in \Omega' \supset \Omega$ be a state in the state space representing the combined system $S+A$, and $\rho_{\text{improper}}$ be the corresponding improper mixed state of $S$\footnote{That is, the corresponding state on $S$ used for experiments on $S$ only.}.   Suppose $\rho_{\text{proper}}= \rho_{\text{improper}}$.   
  Then no experiment on $S$ can distinguish these two cases.  
\end{quote}

We then generalise lemma \ref{Lemma Nonlinear qm violates the MEP} to GPTs:
\begin{lemma}
    \label{Lemma Nonlinear GPTs violates the MEP}
    Let 
    \begin{enumerate}
        \item $\{\omega_i(t_0)\}_i \in \Omega$ be an ensemble of pure states with associated probabilities $\{p_i\}_i, \, \sum_i p_i = 1$, where $\Omega$ is the (convex) state space of the GPT, 
        \item $\rho(t_0) = \sum_i p_i \omega_i(t_0) \in \Omega$,
        \item $T_\rho(t,t_0) = T_\rho(t,t_1)T_\rho(t_1,t_0) : \Omega \to \Omega$ be a time-evolution operator for the state $\rho(t_0)$ such that $\rho(t) = T_\rho(t,t_0)[\rho(t_0)] \in \Omega$, and likewise we write $\omega_i(t) = T_{\omega_i}(t,t_0)[\omega_i(t_0)]$,
        \item $\kappa \in \mathcal{L}(\Omega)$ be an operation on $\Omega$,
        \item $e = \{e_j | \sum_{j \in E} e_j(\rho) = 1\}$ be a measurement on $\rho$, where $e_j : \Omega \to [0,1]$ are effects and $E$ is the set of possible measurement outcomes.
    \end{enumerate}
    Then if 
    \begin{equation}
        \label{Condition nonlinear GPTs violate the MEP}
        \sum_{j\in E} e_j(\kappa(\rho(t)))  \neq \sum_i p_i \sum_{j\in E} e_j(\kappa(\omega_i(t))) \, ,
    \end{equation}
    the MEP is violated \textit{statistically} for a system described by $\rho(t)$ through measurements of $\kappa$.
\end{lemma}

\begin{proof}
    If $\rho_\text{proper}(t_0)$ is a proper mixture of $\{\omega_i(t_0)\}_i$, then it must be so at all times $t$, i.e.
    \begin{equation}
        \rho_\text{proper}(t) = \sum_i p_i \omega_i(t)
    \end{equation}
    following the same argument as in theorem \ref{Nonlinear qm violates the MEP}, with
    \begin{equation}
        \sum_{j\in E} e_j(\kappa(\rho_\text{proper}(t))) = \sum_i p_i \sum_{j\in E} e_j(\kappa(\omega_i(t)))
    \end{equation}
    by lemma \ref{Expval proper mixed states GPTs}.
    
    On the other hand, if $\rho(t_0) = \sum_i p_i \omega_i(t_0)$ is not a proper mixture of $\{\omega_i(t_0)\}_i$, then in general its time evolution is given by
    \begin{equation}
        \rho_\text{improper}(t) = T_\rho(t,t_0)\Big[\sum_i p_i \omega_i(t_0)\Big]
    \end{equation}
    so that the expectation value of measurements of $\kappa$ on $\rho(t)$ is
    \begin{multline}
        \sum_{j\in E} e_j(\kappa(\rho_\text{improper}(t))) \\ = \sum_{j\in E} e_j(\kappa(T_\rho(t,t_0)\Big[\sum_i p_i \omega_i(t_0)\Big]))
    \end{multline}

    Thus, provided condition \eqref{Condition nonlinear GPTs violate the MEP}, the MEP is violated \textit{statistically} for a system described by $\rho(t)$ through measurements of $\kappa$.
\end{proof}

In particular, if $T_\rho(t,t_0)$ is not linear then the MEP is violated for some $\kappa$ and $e_j$ - again, under the constraints of condition \eqref{Condition nonlinear GPTs violate the MEP}, which may not be satisfied if e.g. $\forall i, \, \omega_i(t) \in \ker(\kappa)$ and $\rho(t) \in \ker(\kappa)$. We may now extend theorem \ref{Nonlinear qm violates the MEP} to GPTs:
\begin{theorem}
    Let
    \begin{enumerate}
        \item $\{\omega_i(t_0)\}_i \in \Omega$ be an ensemble of pure states with associated probabilities $\{p_i\}_i, \, \sum_i p_i = 1$, where $\Omega$ is the (convex) state space of the GPT, 
        \item $\rho(t_0) = \sum_i p_i \omega_i(t_0) \in \Omega$,
        \item $T_\rho(t,t_0) = T_\rho(t,t_1)T_\rho(t_1,t_0) : \Omega \to \Omega$ be a time-evolution operator for the state $\rho(t_0)$ such that $\rho(t) = T_\rho(t,t_0)[\rho(t_0)] \in \Omega$, and likewise we write $\omega_i(t) = T_{\omega_i}(t,t_0)[\omega_i(t_0)]$.
    \end{enumerate}
    Then if 
    \begin{equation}
        \rho(t)  \neq \sum_i p_i \omega_i(t) \, ,
    \end{equation}
    the MEP is violated for a system described by $\rho(t)$.
\end{theorem}

\begin{proof}
    Let $\chi(t) := \sum_i p_i \omega_i(t)$. If $\rho(t) \neq \chi(t)$ then $\exists \kappa^*(e_i) \in \mathcal{L}(\Omega,\mathbb{R})$ such that $\kappa^*(e_i)(\rho(t)) \neq \kappa^*(e_i)(\chi(t))$, so $e_i(\kappa(\rho(t))) \neq e_i(\kappa(\sum_i p_i \omega_i(t)))$. Thus, by lemma \ref{Lemma Nonlinear GPTs violates the MEP}, the MEP is violated through such measurements.
\end{proof}

\section{Derivations of Hawking radiation in the Euclidean picture}

\label{app:Hawking Euclidean}

Below, we recall two derivations of the Hartle-Hawking Gibbs state in the Euclidean picture, both of which yield an improper mixture.

\subsection{à la Unruh-Tolman}

\label{app:Unruh-Tolman}

For completeness, we here summarize and paraphrase the derivation of Hawking radiation in Euclidean space presented in \cite{Almheiri2021}.  We stress there is no novelty in the discussions in this and the next section. 

We start with the usual (Minkowskian) Schwarzschild metric
\begin{equation}
    ds^2 = -\Big(1 - \frac{R_S}{r}\Big) c^2 dt^2 + \frac{dr^2}{1-\frac{R_S}{r}} + r^2 d\Omega_2^2
\end{equation}
where $R_S = 2GM/c^2$ is the black hole's Schwarzschild radius. Since the angular directions $d\Omega_2^2$ does not play a role in the following, we omit them for conciseness. The near-horizon metric, obtained through a change of coordinate $r \mapsto R_S(1+\rho^2/4R_S^2), t \mapsto 2 R_S \tau$ in an expansion $\rho \ll R_S$, is the flat Minkowski metric $ds^2 \approx -\rho^2 c^2 d\tau^2 + d\rho^2$ in hyperbolic coordinates. A free-falling observer will thus see nothing special at the horizon $r = R_S$, and the geometry can be analytically extended. Regardless, an observer near the event horizon at fixed $r$ is accelerating to avoid falling in -- in near-horizon coordinates, an observer at fixed $\rho$ is uniformly accelerating in Minkowski spacetime, with proper acceleration $a = \frac{c^2}{\rho}$.

Acceleration in a flat Minkowski spacetime induces an apparent horizon and a Rindler wedge \cite{Earman2011}. The restriction of the Minkowski vacuum state to the subalgebra of field observables in the Rindler wedge yields a so-called KMS state (the algebraic quantum field theory equivalent to a Gibbs state) \cite{Bisognano1975,Bisognano1976}. Such a thermal state is thus improper: it arises from an objective restriction of the available operations and effects available to the observer. Thus, a uniformly accelerating observer experiences thermal radiation -- this is the Unruh effect \cite{Unruh1976}. To obtain the temperature that this observer experiences, one Wick rotates the time coordinate $\theta := -i \tau$ and define $x^0_E := \rho \sin(\theta), x^1 := \rho \cos(\theta)$. The new coordinates $(x^0_E,x^1)$ and $(\rho,\theta)$ are, respectively, Cartesian and polar coordinates on $\mathbb{E}^2$. Now, an observer at constant $\rho$ moves in a circle of circumference $2\pi \rho$.

The partition function in a thermal state is $Z = \Tr[e^{-\frac{\beta}{\hbar} H}]$. Any observable, e.g. $\Tr[\mathcal{O}(\tau) \mathcal{O}(0) e^{-\frac{\beta}{\hbar} H}]$, is periodic under $\tau \mapsto \tau + i \beta$ (or, in imaginary time, under $\theta \mapsto \theta + \beta$), since $\mathcal{O}(\tau) = e^{\frac{i}{\hbar}H\tau} \mathcal{O} e^{-\frac{i}{\hbar}H\tau}$ and using the cyclic property of the trace. Thus, $c \beta$ is the length of the Euclidean time evolution, performed over a circle in Euclidean space. The temperature associated to the partition function is $k_B T = \hbar/\beta$, hence the temperature that an accelerated observer at constant $\rho$ feels is
\begin{equation}
    T_\text{proper} = \frac{\hbar}{k_B \beta} = \frac{c \hbar}{2\pi k_B \rho} = \frac{\hbar}{k_B c}\frac{a}{2\pi} \, .
\end{equation}
This is the temperature of the Unruh effect perceived by a uniformly accelerating thermometer \cite{Unruh1976}. It is however felt by an observer close to the horizon, and decreases as we move away from the black hole. This decrease in temperature is consistent with thermal equilibrium in the presence of a gravitational potential. For a spherically symmetric configuration, the temperature obeys the Tolman relation \cite{Tolman1930}
\begin{equation}
    T_\text{proper}(r) \sqrt{-g_{tt}(r)} = \frac{1}{2\pi} \, .
\end{equation}
This formula is also valid in the full (Lorentzian) geometry, so we can use it to find the temperature that an observer would experience far from the black hole horizon. In this case, we have $g_{tt} = -1$ (Schwarzschild is asymptotically Minkowski) and go to large $r \gg R_S$, so that
\begin{equation}
    T_{H} = T_\text{proper}(r \gg R_S) = \frac{c\hbar}{4\pi k_B R_S} = \frac{\hbar c^3}{8\pi G M k_\text{B}}
\end{equation}
which recovers the Hawking temperature.

\subsection{Euclidean path integral derivation}

\label{app:Euclidean path integral}

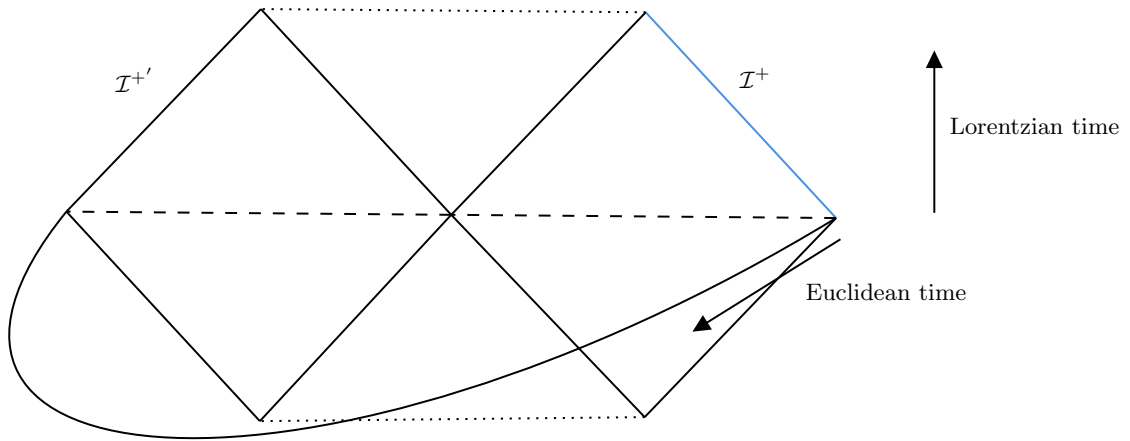
\begin{figure*}
    \centering
    \begin{tikzpicture}[x=0.75pt,y=0.75pt,yscale=-1,xscale=1]
\draw [color={rgb, 255:red, 74; green, 144; blue, 226 }  ,draw opacity=1 ]   (388.79,30.44) -- (484.79,134.29) ;
%Straight Lines [id:da7326873278859336] 
\draw    (194.71,28.82) -- (290.7,132.67) ;
%Straight Lines [id:da221172180351217] 
\draw    (290.7,132.67) -- (388.79,30.44) ;
%Straight Lines [id:da5854076336685918] 
\draw    (96.62,131.05) -- (194.71,28.82) ;
%Straight Lines [id:da09592945738003977] 
\draw    (194.18,236.81) -- (96.62,131.05) ;
%Straight Lines [id:da3559563560810126] 
\draw    (388.26,234.78) -- (290.7,132.67) ;
%Straight Lines [id:da25050425701201307] 
\draw    (290.7,132.67) -- (194.18,236.81) ;
%Straight Lines [id:da6655233732331585] 
\draw    (484.79,134.29) -- (388.26,234.78) ;
%Straight Lines [id:da22091941057296172] 
\draw  [dash pattern={on 0.84pt off 2.51pt}]  (388.26,234.78) -- (194.18,236.81) ;
%Straight Lines [id:da9750900993281979] 
\draw  [dash pattern={on 0.84pt off 2.51pt}]  (388.79,30.44) -- (194.71,28.82) ;
%Straight Lines [id:da7249391297213257] 
\draw  [dash pattern={on 4.5pt off 4.5pt}]  (290.7,132.67) -- (96.62,131.05) ;
%Straight Lines [id:da5312246152903262] 
\draw  [dash pattern={on 4.5pt off 4.5pt}]  (484.79,134.29) -- (290.7,132.67) ;
%Curve Lines [id:da1618588685896889] 
\draw    (96.62,131.05) .. controls (-16.33,270.33) and (221,295) .. (484.79,134.29) ;
%Straight Lines [id:da08297387433330128] 
\draw    (534.33,131.67) -- (534.33,52.67) ;
\draw [shift={(534.33,49.67)}, rotate = 90] [fill={rgb, 255:red, 0; green, 0; blue, 0 }  ][line width=0.08]  [draw opacity=0] (8.93,-4.29) -- (0,0) -- (8.93,4.29) -- cycle    ;
%Straight Lines [id:da6609806875189295] 
\draw    (487,145) -- (414.88,190.08) ;
\draw [shift={(412.33,191.67)}, rotate = 327.99] [fill={rgb, 255:red, 0; green, 0; blue, 0 }  ][line width=0.08]  [draw opacity=0] (8.93,-4.29) -- (0,0) -- (8.93,4.29) -- cycle    ;
%Curve Lines [id:da11398603507115457] 
%\draw [color={rgb, 255:red, 65; green, 117; blue, 5 }  ,draw opacity=1 ]   (290.7,132.67) .. controls (339.67,97) and (419,89) .. (484.79,134.29) ;
%Curve Lines [id:da5401441140732774] 
%\draw [color={rgb, 255:red, 65; green, 117; blue, 5 }  ,draw opacity=1 ]   (96.62,131.05) .. controls (145.58,95.38) and (224.91,87.38) .. (290.7,132.67) ;

% Text Node
\draw (468,165.67) node [anchor=north west][inner sep=0.75pt]   [align=left] {Euclidean time};
% Text Node
\draw (540.67,82.33) node [anchor=north west][inner sep=0.75pt]   [align=left] {Lorentzian time};
% Text Node
\draw (434.67,57.07) node [anchor=north west][inner sep=0.75pt]    {$\mathcal{I}^{+}$};
% Text Node
%\draw (382,81.4) node [anchor=north west][inner sep=0.75pt]    {$\Sigma $};
% Text Node
%\draw (183.33,76.73) node [anchor=north west][inner sep=0.75pt]    {$\overline{\Sigma }$};
% Text Node
\draw (120,57.07) node [anchor=north west][inner sep=0.75pt]    {$\mathcal{I}^{+'}$};
\end{tikzpicture}
    \caption{Extended Euclidean Schwarzschild black hole spacetime \cite{Hartman2015}. The Euclidean path integral yields a pure entangled thermofield double state on $\mathcal{I}^{+'} \cup \mathcal{I}^+$. The quantum state on $\mathcal{I}^+$ is thus an improper mixture -- the Hartle-Hawking state.}
    \label{fig:Euclidean Schwarzschild black hole}
\end{figure*}

Again, for completeness, we summarize and paraphrase derivations of the Hartle-Hawking state using Euclidean path integrals provided in \cite{Hartman2015,Almheiri2021}. We start with the analytically continued Euclidean spacetime
\begin{equation}
    ds^2 = \Big(1 - \frac{R_S}{r}\Big) c^2 dt_E^2 + \frac{dr^2}{1 - \frac{R_S}{r}} + r^2 d\Omega_2^2
\end{equation}
where $t_E := -it$ is the imaginary (Euclidean) time, with $t_E = t_E + \beta$ (as for $\theta$ above). This spacetime only has $r>R_S$, i.e. there is no interior, since $r-R_S$ is like the radial coordinate in polar coordinates and $r=R_S$ is the origin; however, it is analytically extended in the sense that the $t_E = 0$ slice of the Euclidean spacetime extends in $\bar{\mathcal{M}}$. This is shown in Figure \ref{fig:Euclidean Schwarzschild black hole}. To avoid a conical singularity at $r=R_S$ one needs to set $c \beta = 4\pi R_S$. This is necessary to implement the Einstein equivalence principle: an observer falling into an evaporating black hole should not see anything special at the horizon.

Sending $t_E \mapsto t_E + \frac{\beta}{2}$ takes one to the other side of the Penrose diagram in the maximal analytic extension, i.e. from $\mathcal{M}$ to $\bar{\mathcal{M}}$ and vice versa. Thus, just like in Rindler space, we get to the other side of the horizon by going half way around the Euclidean circle.

Imaginary-time periodicity implies a temperature, as we have seen previously. That is, we work with a QFT at finite temperature, and the expectation value of observables should be taken with respect to a thermal state. Real-time evolution by $e^{-\frac{i}{\hbar}Ht}$ corresponds to a path integral on a Lorentzian spacetime, while imaginary-time evolution $e^{-\frac{\beta}{\hbar}H}$ is computed using a path integral on a Euclidean geometry. The Euclidean path integral from the hypersurface $\bar{\Sigma} \cup \Sigma$ living in $\mathcal{\bar{M}} \times \mathcal{M}$ to $\mathcal{I}^{+'} \cup \mathcal{I}^+$ on this extended Euclidean black hole geometry can be computed \cite{Hartman2015} and yields an entangled thermofield double state \eqref{Thermofield double state} on $\mathcal{I}^{+'} \cup \mathcal{I}^+$. Thus, the reduced density matrix on $\mathcal{I}^+$ is the improper Hartle-Hawking state \eqref{density matrix spacetime}.

\section{Gravitational backreactions in the large $N$ limit}

\label{app:Gravitational backreactions in the large N limit}

We here recall the argument which states that a semiclassical approximation to quantum gravity is appropriate when quantum fluctuations in the matter fields are suppressed but still overwhelm the fluctuations in the metric.  
As Wald \cite{Wald1984} puts it: 
\begin{quote}
    \enquote{In the context of quantum field theory in curved spacetime, it is natural to postulate that the back-reaction effects of the quantum field on the gravitational field will be governed by the semiclassical Einstein equation,
    \begin{equation}
        \label{eqn:Wald semiclassical EFE}
        G_{ab} = 8\pi \expval{\hat{T}_{ab}}{\Psi} \quad ;
    \end{equation}
    i.e., it is physically possible for the spacetime to be $(M,g_{ab})$ and for the quantum field to be in state $\Psi$ on $(M,g_{ab})$ if and only if equation \eqref{eqn:Wald semiclassical EFE} is satisfied. Actually, equation \eqref{eqn:Wald semiclassical EFE} would \emph{not} be expected to arise as the lowest approximation to a quantum field theory of gravity coupled to a matter field. This is because in the full theory one would expect to have $\expval{\hat{G}_{ab}} = 8\pi \expval{\hat{T}_{ab}}$ hold exactly, where $\hat{G}_{ab}$ is the full Einstein operator and the state implicit in the expectation values now includes the degrees of freedom of the gravitational field. Furthermore, one would expect that $\hat{G}_{ab}$ would be given in terms of the metric operator by the same formula as holds classical, $\hat{G}_{ab} = G_{ab}[\hat{g}_{cd}]$. However, since $G_{ab}$ is a nonlinear function of $g_{cd}$ we expect $\expval{\hat{G}_{ab}} \neq G_{ab}[\expval{\hat{g}_{cd}}]$. Indeed, if we write $\hat{g}_{ab} = g^C_{ab} \hat{I} + \hat{\gamma}_{ab}$ -- where $g^C_{ab}$ is a classical solution of Einstein's equation and $\hat{I}$ is the identity operator -- and if we keep only terms quadratic in $\hat{\gamma}_{ab}$ in the formula for $\hat{G}_{ab}$, then $\expval{\hat{G}_{ab}}$ and $G_{ab}[\expval{\hat{g}_{ab}}]$ will differ by $-8\pi \expval{\hat{t}_{ab}}$, where $\hat{t}_{ab}$ is given in terms of $\hat{\gamma}_{ab}$ by a formula very similar to that of $\hat{T}_{ab}$ in terms of $\hat{\phi}$ [...] and the contribution from this term should be comparable to that of $\expval{\hat{T_{ab}}}$. One can then interpret this fact as saying that the quantum back-reaction effects caused by \emph{gravitons} [...] are as important as that of any other quantum field, and should not be neglected in equation \eqref{eqn:Wald semiclassical EFE}. Nevertheless, one can justify equation \eqref{eqn:Wald semiclassical EFE} in terms of a systematic approximation to a full quantum field theory including gravitation as follows. If we have $N$ matter fields present, then, roughly speaking, the effects of the matter fields will be $N$ times as important as that of gravitons. Hence, in the limit of large $N$, the neglect of gravitons should be justified, and one will obtain equation \eqref{eqn:Wald semiclassical EFE} [...] as the lowest approximation in a \enquote{$1/N$ expansion} of the full theory of quantum gravity coupled to matter. In any case, equation \eqref{eqn:Wald semiclassical EFE} should at least provide a qualitative indication of the back reaction effects produced by quantum fields on the gravitational field.}
\end{quote}

In a $1/N$ expansion of quantum gravity with $N$ matter fields in the semiclassical limit, we indeed have \cite{Lowe2023}
\begin{equation}
    \expval{\hat{T}_{\mu\nu}} = T_{\mu\nu}^C + \sum_{l=1}^{\infty} \hbar^l \expval{\hat{T}_{\mu\nu}^{(l)}} = T_{\mu\nu}^C + \mathcal{O}(N\hbar)
\end{equation}
where $T_{\mu\nu}^C$ is the classical contribution (given by the one point functions $\phi_i^{\text{cl}} = \expval{\hat{\phi}_i}$ only) to the energy-momentum tensor and $\hat{T}_{\mu\nu}^{(l)}$ is the $l$-th order quantum correction to the energy-momentum tensor. Likewise we can make an expansion of the metric
\begin{equation}
    g_{\mu\nu} = g_{\mu\nu}^C + \sum_{n=1}^\infty \hbar^n g_{\mu\nu}^{(n)} \, .
\end{equation}
In the $1/N$ expansion of quantum gravity, one can then substitute in the $g^{(n)}_{\mu\nu}$ with $n<l$ into equation \eqref{eqn:Wald semiclassical EFE} using $\expval{\hat{T}_{\mu\nu}^{(l)}}$ to compute the next higher-order correction $g_{\mu\nu}^{(l)}$ \cite{Lowe2023}. In particular, provided $N\hbar \ll 1$ (while $N \gg 1$ and $\hbar \ll 1$), we have that the leading order corrections to the classical metric $g_{\mu\nu}^C$ in the semiclassical Einstein field equations are from the $\mathcal{O}(N\hbar)$ contributions in the energy-momentum tensor. Following Wald's argument, we then have that
\begin{align}
    &\expval{\hat{G}_{ab}} \stackrel{\mathcal{O}(\hat{\gamma}_{ab}^2)}{\approx} G_{ab}[g^C_{ab}] - 8\pi \underbrace{\expval{\hat{t}_{ab}}}_{\mathcal{O}(\hbar)} = 8\pi\underbrace{\expval{\hat{T}_{ab}}}_{T_{\mu\nu}^C + \mathcal{O}(N\hbar)} \\ \Rightarrow &G_{ab}[g^C_{ab}] \stackrel{\mathcal{O}(N\hbar)}{\approx} 8\pi\expval{\hat{T}_{ab}} + \mathcal{O}(\hat{\gamma}_{ab}^3,\hbar)
\end{align}
Thus, in the limit where $N\hbar \ll 1$ with $N \gg 1$, the semiclassical Einstein field equations are indeed recovered. These contributions are negligible at SI scales, which is why classical GR is so successful. The first quantum mechanical post-GR corrections in a 1/$N$ perturbative expansion, which come into play at Planck scales for large $N$, yield the semiclassical Einstein field equations to that order.

\bibliographystyle{unsrt}
\bibliography{library}

\end{document}